\documentclass[letterpaper]{llncs}

\usepackage{times,amssymb,amsmath,verbatim,booktabs}
\newtheorem{observation}{Observation}
\usepackage[ruled,vlined,linesnumbered]{algorithm2e}
\usepackage{graphicx}

\newcommand{\cP}{{\mathcal P}}
\newcommand{\cA}{{\mathcal A}}

\newcommand{\cR}{{\mathcal R}}

\newboolean{short}
\setboolean{short}{false}
\newcommand{\shortOnly}[1]{\ifthenelse{\boolean{short}}{#1}{}}
\newcommand{\onlyShort}[1]{\ifthenelse{\boolean{short}}{#1}{}}
\newcommand{\longOnly}[1]{\ifthenelse{\boolean{short}}{}{#1}}
\newcommand{\onlyLong}[1]{\ifthenelse{\boolean{short}}{}{#1}}
\newcommand{\shortLong}[2]{\ifthenelse{\boolean{short}}{#2}{#1}}
\newcommand{\longShort}[2]{\ifthenelse{\boolean{short}}{#2}{#1}} 
\onlyShort{

\usepackage{enumitem}
\setitemize{noitemsep,topsep=0pt,parsep=0pt,partopsep=0pt}
\setenumerate{noitemsep,topsep=0pt,parsep=0pt,partopsep=0pt}
}

\begin{document}
\mainmatter
\title{Byzantine Geoconsensus}
\author{Joseph Oglio \and Kendric Hood \and Gokarna Sharma \and Mikhail Nesterenko}
\institute{Department of Computer Science, Kent State University, Kent, OH 44242, USA\\
\email{\{joglio@,khood@,sharma@cs.,mikhail@cs.\}kent.edu}}

\maketitle



\begin{abstract}
We define and investigate the consensus problem for a set of $N$ processes embedded on the $d$-dimensional plane, $d\geq 2$, which we call the {\em geoconsensus} problem. The processes have unique coordinates and can communicate with each other through oral messages. In contrast to the literature where processes are individually considered Byzantine, it is considered that all processes covered by a finite-size convex fault area $F$ are Byzantine and there may be one or more processes in a fault area. Similarly as in the literature where correct processes do not know which processes are Byzantine, it is assumed that the fault area location is not known to the correct processes. We prove that the geoconsensus is impossible if all processes may be covered by at most three areas where one is a fault area.

Considering the 2-dimensional embedding, on the constructive side,  for $M \geq 1$  fault areas $F$ of arbitrary shape with diameter $D$,  we present a consensus algorithm that tolerates $f\leq N-(2M+1)$ Byzantine processes provided that there are $9M+3$ processes with pairwise distance between them greater than $D$. For square $F$ with side $\ell$,
we provide a consensus algorithm that lifts this pairwise distance requirement and tolerates $f\leq N-15M$ Byzantine processes given that all processes are covered by at least $22M$ axis aligned squares of the same size as $F$. For a circular $F$ of diameter $\ell$, this algorithm tolerates $f\leq N-57M$ Byzantine processes if all processes are covered by at least $85M$ circles. We then extend these results to various size combinations of fault and non-fault areas as well as $d$-dimensional process embeddings, $d\geq 3$.
\end{abstract}

\section{Introduction}



The problem of {\em Byzantine consensus}~\cite{lamport1982byzantine,Pease80} has been attracting extensive attention from researchers and engineers in distributed systems. 
It has applications in distributed storage~\cite{abd2005fault,adya2002farsite,PBFT,castro2003base,kubiatowicz2000oceanstore}, secure communication~\cite{cramer1997secure}, safety-critical systems~\cite{rushby2001bus}, blockchain~\cite{miller2016honey,sousa2018byzantine,zamani2018rapidchain}, and Internet of Things (IoT) \cite{LaoD0G20}.  

Consider a set of $N$ processes with unique IDs that can communicate with each other. 
Assume that $f$ processes out of these $N$ processes are Byzantine. Assume also that which process is Byzantine is not known to correct processes, except possibly the size $f$ of Byzantine processes. The Byzantine consensus problem here requires the $N-f$ correct processes to reach to an agreement tolerating arbitrary behaviors of the $f$ Byzantine processes.  


Pease {\it et al.}~\cite{Pease80} showed that the maximum possible number of faults $f$ that can be tolerated depends on the way how the (correct) processes communicate: through oral messages or through unforgable written messages (also called signatures). An {\em oral} message is completely under the control of the sender, therefore, if the sender is Byzantine, then it can transmit any possible message. This is not true for a signed, written message.  
Pease {\it et al.}~\cite{Pease80} showed that the consensus is solvable only if $f < N/3$ when communication between processes is through oral messages. For signed, written messages, they showed that the consensus is possible tolerating any number of faulty processes $f\leq N$.


The Byzantine consensus problem discussed above assumes nothing about the locations of the processes, except that they have unique IDs. Since each process can communicate with each other, it can be assumed that  the $N$ processes work under a complete graph (i.e., clique) topology consisting of $N$ vertices and $N(N-1)/2$ edges.  
Byzantine consensus has also been studied in arbitrary graphs~\cite{vaidya2012iterative,Pease80} and in wireless networks~\cite{moniz2012byzantine}, relaxing the complete graph topology requirement so that a process may not be able to communicate with all other $N-1$ processes. The goal in these studies is to establish necessary and sufficient conditions for consensus to be solvable. For example,  Pease {\it et al.}~\cite{Pease80} showed that the consensus is solvable through oral messages tolerating $f$ Byzantine processes if the communication topology is $3f$-regular. 
Furthermore, there is a number of studies on a related problem of {\em Byzantine broadcast} when the communication topology is not a complete graph topology, see for example~\cite{koo2004broadcast,pelc2005broadcasting}. Byzantine broadcast becomes fairly simple for a complete graph topology. 

Recently, motivated by IoT-blockchain applications, Lao {\it et al.} \cite{LaoD0G20} proposed a consensus protocol, which they call Geographic-PBFT or simply G-PBFT, that extends the well-known PBFT consensus protocol by Castro and Liskov \cite{PBFT} to the geographic setting. 
The authors considered the case  of fixed IoT devices embedded on geographical locations for data collection and processing. The location data can be obtained through recording location information at the installation time or can also be obtained using low-cost GPS receivers or location estimation algorithms~\cite{bulusu2004self,hightower2001location}.
They argued that the fixed IoT devices have more computational power than other mobile IoT devices (e.g., mobile phones and sensors) and are less likely to become malicious nodes.
They then exploited (geographical) location information of fixed IoT devices to reach consensus. They argued that G-PBFT avoids Sybil attacks, reduces the overhead for validating and recording transactions, and achieves high consensus efficiency and low traffic intensity. However, G-PBFT is validated only experimentally and no formal analysis is given. 

In this paper, we formally define and study the Byzantine consensus problem when processes are embedded on the geographical locations in fixed unique coordinates, which we call the {\em Byzantine geoconsensus} problem. If fault locations are not constrained, the geoconsensus problem differs little from the Byzantine consensus. This is because the unique locations serve as IDs of the processes and same set of results can be established depending on whether communication between processes is through oral messages or unforgable written messages. 
Therefore, we relate the fault locations to the geometry of the problem, assuming that the faults are limited to a \emph{fault area} $F$ (going beyond the limitation of mapping Byzantine behavior to individual processes). In other words, the fault area lifts the restriction of mapping Byzantine behavior to individual processes in the classic setting and now maps the Byzantine behavior to all the processors within a certain area in the geographical setting.  Applying the classic approaches of Byzantine consensus may not exploit the collective Byzantine behavior of the processes in the fault area and hence they may not provide benefits in the geographical setting.   
Furthermore,
we are not aware of prior work in Byzantine consensus where processes are embedded in a geometric plane while faulty processes are located in a fixed area.

In light of the recent development on location-based consensus protocols, such as G-PBFT \cite{LaoD0G20}, discussed above, 
we believe that our setting deserves a formal study. 
In this paper we consider the Byzantine geoconsensus problem in case the processes are embedded in a $d$-dimensional plane, $d\geq 2$.   
We study the possibility and bounds for a solution to geoconsensus. 
We demonstrate that geoconsensus allows quite robust solutions: all but a fixed number of processes may be Byzantine. 

\vspace{1mm}
\noindent\textbf{Contributions.} Let $N$ denotes the number of processes, $M$ denotes the number of fault areas $F$, $D$ denotes the diameter of $F$, and $f$ denotes the number of faulty processes. Assume that each process can communicate with all other $N-1$ processes and the communication is through oral messages. Assume that all the processes covered by a faulty area $F$ are Byzantine. 
The correct processes know the size of each faulty area (such as its diameter, number of edges, area, etc.) and the total number $M$ of them but do not know their exact location. 

In this paper, we made the following five contributions:
\onlyShort{\vspace{-1mm}}
\begin{itemize}
\item [(i)] 
An impossibility result 
that geoconsensus is not solvable if all $N$ processes may be covered by $3$ equal size areas $F$ and one of them may be fault area. This extends to the case of $N$ processes being covered by $3M$ areas $F$ with $M$ areas being faulty.

\item [(ii)] The algorithm {\em BASIC} that 
solves geoconsensus tolerating $f\leq N-(2M+1)$ Byzantine processes, provided that there are $9M+3$ processes with pairwise distance between them greater than $D$. 

\item [(iii)] The algorithm {\em GENERIC} that 
solves geoconsensus tolerating $f\leq N-15M$ Byzantine processes, provided that all $N$ processes are covered by $22M$ axis-aligned squares of the same size as the fault area $F$, removing the pairwise distance assumption in the algorithm {\em BASIC}. 

\item [(iv)] An extension of the {\em GENERIC} algorithm to circular $F$ tolerating $f\leq N-57M$ Byzantine processes if all $N$ processes are covered by $85M$ circles of same size as $F$.

\item [(v)] Extensions of the results (iii) and (iv) to various size combinations of fault and non-fault areas as well as to $d$-dimensional process embeddings, $d\geq 3$. 
\end{itemize}

\onlyShort{\vspace{-1mm}}
Our results are interesting as they provide trade-offs among $N, M,$ and $f$, which is in contrast to the trade-off provided only between $N$ and $f$ in the Byzantine consensus literature. For example, the results in Byzantine consensus show that only $f<N/3$ Byzantine processes can be tolerated, whereas our results show that as many as $f\leq N-\alpha M$, Byzantine processes can be tolerated provided that the processes are placed on the geographical locations so that at least $\beta M$ areas (same size as $F$) are needed to cover them. Here $\alpha$ and $\beta$ are both integers with $\beta\geq c\cdot \alpha$ for some constant $c$. 

Furthermore, our geoconsensus algorithms reduce the message and space complexity in solving consensus. In the Byzantine consensus literature, every process sends communication with every other process in each round. Therefore, in one round there are $O(N^2)$ messages exchanged in total.
As the consensus algorithm runs for $O(f)$ rounds, in total $O(f\cdot N^2)$ messages are exchanged in the worst-case.
In our algorithms, let $N$ processes are covered by $X$ areas of size the same as fault area $F$. Then in a round
only $O(X^2)$ messages are  exchanged. Since the algorithm runs for $O(M)$ rounds to reach geoconsensus, in total $O(M\cdot X^2)$ messages are exchanged in the worst-case. Therefore, our geoconsensus algorithms are message (equivalently communication) efficient. The improvement on space complexity can also be argued analogously. 

Finally, Pease {\it et al.} \cite{Pease80} showed that it is impossible to solve consensus through oral messages when $N=3f$ but there is a solution when $N\geq 3f+1$. That is, there is no gap on the impossibility result and a solution. We can only show that it is impossible to solve consensus when all $N$ processes are covered by $3M$ areas that are the same size as $F$ but there is a solution when all $N$ processes are covered by at least $22M$ areas (for the axis-aligned squares case). Therefore, there is a general gap between the condition for impossibility and the condition for a solution. We leave this gap as open and note that it would be interesting at the same time challenging to close this gap.   

\vspace{1mm}
\noindent\textbf{Techniques.}
Our first contribution is established extending the impossibility proof technique of Pease {\it et al.} \cite{Pease80} for Byzantine consensus to the geoconsensus setting. 
The algorithm {\em BASIC} is established first through a leader selection to compute a set of leaders so that they are pairwise more than distance $D$ away from each other and then running carefully the Byzantine consensus algorithm of Pease {\it et al.} \cite{Pease80} on those leaders. 

For the algorithm {\em GENERIC},
we start by covering processes by axis-aligned squares and studying how these squares may intersect with fault areas of various shapes and sizes. Determining optimal axis-aligned square coverage is NP-hard. We provide constant-ratio approximation algorithms. We also discuss how to cover processes by circular areas. 
Then, we use these ideas to construct algorithm {\em GENERIC}  for fault areas that are either square or circular, which does not need the pairwise distance requirement of \emph{BASIC} but requires the bound on the number of areas in the cover area set. Finally, we extend these ideas to develop covering techniques for higher dimensions. 


\vspace{1mm}
\noindent\textbf{Roadmap.} We introduce notation and the geoconsensus problem, and 
establish an impossibility of geoconsensus in  Section~\ref{section:notation}. 
%
We present in Section \ref{section:simple} algorithm {\em BASIC}. 
We discuss covering processes 
in Section \ref{section:covering}.
We then present in Section \ref{section:generic} algorithm {\em GENERIC}. 
In Section~\ref{section:extension}, we extend the results to $d$-dimensional space, $d\geq 3$. In Section~\ref{section::end}, we conclude the paper with a short discussion on future work. 
\onlyShort{Pseudocodes, some details, and many proofs are deferred to Appendix due to space constraints.}

\onlyShort{\vspace{-3mm}}
\section{Notation, Problem Definition, and Impossibility}
\label{section:notation}
\vspace{-2mm}
\noindent{\bf Processes.} A computer system consists of a set $\cP=\{p_1,\ldots,p_N\}$ of $N$ processes. Each process $p_i$ embedded in the 2-dimensional plane has unique planar coordinates $(x_i,y_i)$. 
The coordinates for higher dimensions can be defined accordingly.  
Each process is aware of coordinates of all the other processes and is capable of sending a message to any of them. The sender of the message may not be spoofed. Communication is synchronous. The communication between processes is through oral messages. 

\vspace{1mm}
\noindent{\bf Byzantine faults.}
A process may be either correct or faulty. The fault is Byzantine. A faulty process may behave arbitrarily. This fault is permanent. To simplify the presentation, we assume that all faulty processes are controlled by a unique adversary trying to thwart the system from achieving its task.

\vspace{1mm}
\noindent{\bf Fault area.} 
The adversary controls the processes as follows. Let the \emph{fault area} $F$ be a finite-size convex area in the plane. Let $D$ be the diameter of $F$, i.e. the maximum distance between any two points of $F$. The adversary may place $F$ in any location on the plane. A process $p_i$ is {\em covered} by $F$ if the coordinate $(x_i,y_i)$ of $p_i$ is either in the interior or on the boundary of $F$.  
Every process covered by $F$ is faulty.

\onlyLong{
\begin{table}[!t]
\vspace{-4mm}
{\footnotesize
\centering
\begin{tabular}{l|l}
\toprule
{\bf Symbol} & {\bf Description}  \\
\toprule
$N$; $\cP$; $(x_i,y_i)$ &  number of processes; $\{p_1,\ldots,p_N\}$; planar coordinates of process $p_i$\\
\hline
$F; D$; $\mathcal{F}$ & fault area; diameter of $F$; a set of fault areas $F$ with $|\mathcal{F}|=M$\\
\hline
$f$ & number of faulty processes\\
\hline
$\cP_D$ & processes in $\cP$ such that pairwise distance between them is more than $D$\\
\hline
$A$ (or $A_j(R_i)$); $\cA$ & cover area that is of same shape and size as $F$; a set of cover areas $A$\\
\hline
$n(F)$ & number of cover areas $A\in \cA$ that a fault area $F$ overlaps\\ 
\bottomrule
\end{tabular}
\caption{Notation used throughout the paper.}
\label{table:1}
\vspace{-2mm}
}
\end{table}
}

A \emph{fault area set} or just \emph{fault set} is the set
$\mathcal{F}$
of identical fault areas $F$. The size of this set is $M$, i.e., $|\mathcal{F}|=M$. The adversary controls the placement of all areas in $\mathcal{F}$. 
Correct processes know the shape and  size of the fault areas $F$ as well as $M$, the size of $\mathcal{F}$. However, correct processes do not know the precise placement of the fault areas $\mathcal{F}$. 
For example, if $\mathcal{F}$ contains $4$ fault square fault areas $F$ with the side $\ell$, then correct processes know that there are $4$ square fault areas with side $\ell$ each but do not know where they are located.
Table \ref{table:1} \onlyShort{(in Appendix)} summarizes notation used in this paper.

\vspace{1mm}
\noindent{\bf Byzantine Geoconsensus.}  
Consider the binary consensus where every correct process is input a value $v \in \{0, 1\}$ and must output an irrevocable decision with the following three properties.

\onlyShort{\vspace{-4mm}}
\begin{samepage}
\begin{description}
\item[agreement] -- no two correct processes decide differently;
\item[validity] -- if all the correct processes input the same value $v$, then every correct process decides $v$;
\item[termination] -- every correct process eventually decides.
\end{description}
\end{samepage}
\onlyShort{\vspace{-4mm}}
\begin{definition} An algorithm solves \emph{the Byzantine geoconsensus Problem} (or \emph{geoconsensus} for short) for fault area set $\mathcal{F}$, if every computation produced by this algorithm satisfies the three consensus properties. 
\end{definition}

\noindent{\bf Impossibility of Geoconsensus.}
\label{section:impossibility}
Given a certain set of embedded processes $\cP$ and single area $F$, the \emph{coverage number} $k$
of $\cP$ by $F$ is the minimum number of such areas required to cover each node of $\cP$. We show that geoconsensus is not solvable if the coverage number $k$ is less than $4$. 
When the coverage number is $3$ or less, the problem translates into classic consensus with 3 sets of peers where one of the sets is faulty. Pease {\it et al.}~\cite{Pease80} proved the solution to be impossible. The intuition is that a group of correct processes may not be able to distinguish which of the other two groups is Byzantine and which one is correct. Hence, the correct groups may not reach consensus. 

\vspace{-1mm}
\begin{theorem}[Impossibility of Geoconsensus]
\label{theorem:impossibility}
Given a set $\cP$ of $N\geq 3$ processes and an area $F$,
there exists no algorithm that solves the geoconsensus Problem if the coverage number $k$ of $\cP$ by $F$ is less than $4$.
\end{theorem}

\onlyLong{
\begin{proof}
Set $N=3\cdot \kappa$, for some positive integer $\kappa\geq 1$. Place three areas $A$ on the plane in arbitrary locations.  To embed processes in $\cP$, consider a bijective placement function $f:\cP\rightarrow\cA$ such that $\kappa$ processes are covered by each area $A$. Let $v$ and $v'$ be two distinct input values $0$ and $1$. Suppose one area $A$ is fault area, meaning that all $\kappa$ processes in that area are faulty. 

This construction reduces the Byzantine goeconsensus problem to the impossibility construction for the classic Byzantine consensus problem given in the theorem in Section 4 of Pease {\it et al.}~\cite{Pease80} for the $3\kappa$ processes out of which $\kappa$ are Byzantine.
\qed
\end{proof}
}

\onlyLong{
\begin{algorithm}[!t]
{\small
{\bf Setting:} A set $\cP$ of $N$ processes positioned at distinct coordinates. Each process can communicate with all other processes and knows their coordinates.  There are $M\geq 1$ identical fault areas $F$. The diameter of a fault area is $D$. The locations of any area $F$ is not known to correct processes. Each process covered by any $F$ is Byzantine.\\
{\bf Input:} Each process has initial value either 0 or 1.\\
{\bf Output:} Each correct process outputs decision subject to Geoconsensus.\\
{\bf \em Procedure for process $p_k\in \cP$}\\
// leaders selection \\
Let $P_D \leftarrow \emptyset$, $P_C  \leftarrow \cP$;  \\  
\While{$P_C \neq \emptyset$}{
    let $P_3 \subset P_D$ be a set of processes such that $\forall p_j \in P_3$, 
        $\mathit{Nb}(p_j, D)$ has distance $D$ independent set of at most 3; \\
    let $p_i \in P_3$, located in $(x_i, y_i)$ be
    the lexicographically smallest process in $P_3$, i.e. $\forall p_j \neq p_i \in P_3:$ located in $(x_j, y_j)$ either  $x_i<x_j$ or $x_i=x_j$ and $y_i<y_j$; \\
    add $p_i$ to $P_D$; \\
    remove $p_i$ from $P_C$; \\
   $\forall p_j \in \mathit{Nb}(p_i, D)$ remove $p_j$ from $P_C$; \\
}

// consensus \\
\eIf{$p_k \in P_D$}
{run \emph{PSL} algorithm, achieve decision $v$, broadcast $v$, output $v$;}
{
wait for messages with identical decision $v$ from at least $2M + 1$ processes from $\cP_D$, output $v$; 
}

\caption{\emph{BASIC} geoconsensus algorithm.}
\label{algorithm:simple}
}
\end{algorithm}	

}

\section{The BASIC Geoconsensus Algorithm}
\label{section:simple}
In this section, we present the algorithm we call \emph{BASIC} that solves geoconsensus for up to
$f<N-(2M+1),M\geq 1$ faulty processes located in fault area set $\mathcal{F}$ of size $|\mathcal{F}|=M$ provided that $\cP$ contains at least $9M+3$ processes such that the pairwise distance between them is greater than the diameter $D$ of the fault areas $F\in  \mathcal{F}$.

The pseudocode of \emph{BASIC} is shown in Algorithm \ref{algorithm:simple}\onlyShort{~(in Appendix)}. 
It contains two parts: the leaders selection and the consensus procedure. 
The first component is the selection of leaders. So as to not be covered by $F$ jointly, the leaders need to be located pairwise distance more than $D$ away from each other. Finding the largest set of such leaders is equivalent to computing the maximal independent set in a unit disk graph. This problem is known to be NP-hard~\cite{clark1990unit}. We, therefore, employ a greedy heuristic. 

For the leaders selection, for each process $p_i$, denote by $\mathit{Nb}(p_i,D)$ the distance $D$ neighborhood of $p_i$. That is, $p_j \in \mathit{Nb}(p_i,D)$ if $d(p_i, p_j) \leq D$. A distance $D$ independent set for a planar graph is a set of processes such that all processes in the planar graph are at most $D$ away from the processes in this independent set. It is known~\cite[Lemma 3.3]{marathe1995simple} that every distance $D$ graph has a neighborhood whose induced subgraph contains any independent set of size at most 3. 

The set of leaders $\cP_D \subset \cP$ selection procedure operates as follows. A set $P_C$ of leader candidates is processed. At first, all processes are candidates. All processes whose distance $D$ neighborhood induce a subgraph with an independent set no more than $3$ are found. The process $p_i$ with lexicographically smallest coordinates, i.e. the process in the bottom left corner, is selected into the leader set $\cP_D$. Then, all processes in $\mathit{Nb}(p_i,D)$ are removed from the leader candidate set $\cP_C$. This procedure repeats until $\cP_C$ is exhausted.

The second part of \emph{BASIC} relies on the classic consensus algorithm of Pease {\it et al.}~\cite{Pease80}. We denote this algorithm as \emph{PSL}. The input of \emph{PSL} is the set of $3f+1$ processes such that at most $f$ of them are faulty as well as the input $1$ or $0$ for each process. As output, the correct processes provide the decisions value subject to the three properties of the solution to consensus.  \emph{PSL} requires $f+1$ communication rounds. 

The complete \emph{BASIC} operates as follows. All processes select leaders in $P_D$. Then, the leaders run \emph{PSL} and broadcast their decision. The rest of the correct processes, if any, adopt this decision.

\vspace{2mm}
\noindent{\bf Analysis of \emph{BASIC}.}
The observation below is immediate since all processes run exactly the same deterministic leaders selection procedure. 
\begin{observation}
For any two processes $p_i,p_j\in \cP$,  set $P_D$ computed by $p_i$ is the same as set $P_D$ computed by $p_j$.
\end{observation}

\begin{lemma}\label{lemPD}
If $\cP$ contains at least $3x$ processes such that the distance between any pair of such processes is  $>D$, then the size of $P_D$ computed by processes in BASIC is $\geq x$.
\end{lemma}
\begin{proof}
For the same problem like ours, in~\cite[Theorem 4.7]{marathe1995simple}, it is proven that the heuristic we use for the leaders selection provides a distance $D$ independent set $P_D$ whose size is no less than a third of optimal size. Thus, $x \leq |P_D|$. The lemma follows.
\qed
\end{proof}

\begin{lemma}
\label{lemma:no-two-process-in-one-area}
Consider a fault area $F$ with diameter $D$. No two processes in $\cP_D$ are covered by $F$. 
\end{lemma}
\begin{proof}
For any two processes $p_i,p_j\in \cP_D$, $d(p_i,p_j)>D$. Since any area $F$ has diameter $D$, no two processes $>D$ away can be covered by $F$ simultaneously.
\qed
\end{proof}

\begin{theorem} 
\label{theorem:basic}
Algorithm BASIC solves the Byzantine geoconsensus Problem for a fault area set $\mathcal{F}$, the size of $M \geq 1$  with fault areas $F$ with diameter $D$ for $N$ processes in $\cP$ tolerating $f\leq N-(2M+1)$ Byzantine faults provided that $\cP$ contains at least $9M + 3$ processes such that their pairwise distance is more than $D$. The solution is achieved in $M + 2$ communication rounds. 
\end{theorem}
\begin{proof} If $\cP$ contains at least $9M +3$ processes whose pairwise distance is more than $D$, then, according to Lemma~\ref{lemPD}, each processes in \emph{BASIC} selects $P_D$ such that $|P_D|\geq 3M + 1$.
We have $M\geq 1$ fault areas, i.e., $|\mathcal{F}|=M$.  From Lemma \ref{lemma:no-two-process-in-one-area}, a process $p\in \cP_D$ can be covered by at most one fault area $F$. Therefore, when $|P_D|\geq 3M+1$, then it is guaranteed that even when $M$ processes in $\cP_D$ are Byzantine, $2M+1$ correct processes in $\cP_D$ can reach consensus using \emph{PSL} algorithm. 

In the worst case, the adversary may position fault areas of $\mathcal{F}$ such that all but $2M+1$ processes in $\cP$ are covered. Hence, \emph{BASIC} tolerates $N - (2M+1)$ faults. 

Let us address the number of rounds that \emph{BASIC} requires to achieve geoconsensus. It has two components executed sequentially: leaders election and \emph{PSL}. Leaders election is done independently by all processes and requires no communication. \emph{PSL}, takes $M+1$ rounds for the $2M + 1$ leaders to arrive at the same decision. It takes another round for the leaders to broadcast their decision. Hence, the total number of rounds is $M + 2$.
\qed
\end{proof}

\onlyShort{\vspace{-3mm}}
\section{Covering Processes}
\label{section:covering}
\vspace{-2mm}
In this section, in preparation for describing the {\em GENERIC} geoconsensus algorithm, we discuss techniques of covering processes by axis-aligned squares and circles. These techniques vary depending on the shape and alignment of the fault area $F$.  

\vspace{1mm}
\noindent{\bf Covering by Squares.}
\label{subsection:square-covers}
The algorithm we describe below covers the processes by square areas $A$ of size $\ell\times\ell$, assuming that the fault areas $F$ are also squares of the same size.  
Although $F$ may not be axis-aligned, we use axis-aligned areas $A$ for the cover and later determine how many such axis-aligned areas $A$ that possibly non-axis-aligned fault area $F$ may overlap. 
Let $A$ be positioned on the plane such that the coordinate of its bottom left corner is $(x_1,y_1)$. The coordinates of its top left, top right, and bottom right corners are respectively $(x_1,y_1+\ell),(x_1+\ell,y_1+\ell),$ and $(x_1+\ell,y_1)$.

Let process $p_i$ be at coordinate $(x_i,y_i)$. We say that $p_i$ is {\em covered} by $A$ if and only if $x_1\leq x_i\leq x_1+\ell$ and $y_1\leq y_i\leq y_1+\ell$.  We assume that $A$ is {\em closed}, i.e., process $p_i$ is assumed to be covered by $A$ even when $p_i$ is positioned on the boundary of $A$. 

We first formally define the covering problem by square areas $A$, which we denote by SQUARE-COVER. Let $\cA$ be a set of square areas $A$. We say that $\cA$ completely covers all $N$ processes if each $p_i\in \cP$ is covered by at least one square of $\cA$. 

\begin{definition}[The  SQUARE-COVER problem]
\label{definition:square-cover}
Suppose $N$ processes are embedded into a 2d-plane such that the coordinates of each process are unique. The SQUARE-COVER problem  is to determine if a certain number of square areas $A=\ell \times \ell$ can completely cover these $N$ processes.  
\end{definition}

\begin{theorem}
\label{theorem:square-cover}
SQUARE-COVER is NP-Complete.
\end{theorem}

\onlyLong{
\begin{proof}
The proof is to show that SQUARE-COVER is equivalent to the BOX-COVER problem which was shown to be NP-Complete by Fowler {\it et al.}~\cite{Fowler81}.
BOX-COVER is defined as follows: There is a set of $N$ points on the plane such that each point has unique integer coordinates. A closed box (rigid but relocatable) is set to be a square with side 2 and is axis-aligned. The problem is to decide whether a set of $k\geq 1$ identical axis-aligned closed boxes are enough to completely cover all $N$ points. 
Fowler {\it et al.} provided a polynomial-time reduction of 3-SAT to BOX-COVER such that $k$ boxes will suffice if and only if the 3-SAT formula is satisfiable.  In this setting,  SQUARE-COVER (Definition \ref{definition:square-cover}) reduces to BOX-COVER for $\ell=2$. Therefore, the NP-Completeness of BOX-COVER extends to SQUARE-COVER. \qed
\end{proof}
}

\noindent{\bf A Greedy Square Cover Algorithm.}
\label{subsection:greedy}
Since SQUARE-COVER is NP-Complete, we use a greedy approximation algorithm to find a set $\cA$ of $k_{greedy}$ axis-aligned square areas  $A=\ell\times \ell$ that completely cover all $N$ processes in $\cP$. We prove that  $k_{greedy}\leq 2\cdot k_{opt}$ (i.e.,  2-approximation), where $k_{opt}$ is the optimal number of axis-aligned squares in any algorithm to cover those $N$ processes. We call this algorithm \emph{GSQUARE}. Each process $p_i$ can run \emph{GSQUARE} independently, because $p_i$ knows all required input parameters for \emph{GSQUARE}. 

\emph{GSQUARE} operates as follows. Suppose the coordinates of process $p_i\in \cP$ are
$(x_i,y_i)$. Let $x_{min}=\min_{1\leq i\leq N}x_i, x_{max}=\max_{1\leq i\leq N}x_i, y_{min}=\min_{1\leq i\leq N}y_i,$ and $y_{max}=\max_{1\leq i\leq N}y_i.$ 
Let $R$ be an axis-aligned rectangle with the bottom left corner at $(x_{min},y_{min})$ and the top right corner at $(x_{max},y_{max})$. 
It is immediate that $R$ is the smallest axis-aligned rectangle that covers all $N$ processes. 
The width of $R$ is $width(R)=x_{max}-x_{min}$ and the height is $height(R)=y_{max}-y_{min}$. See Figure~\ref{fig:rectangle} for illustration.

\begin{figure}[!t]
\vspace{-0mm}
\begin{center}
\includegraphics[height=2.00in]{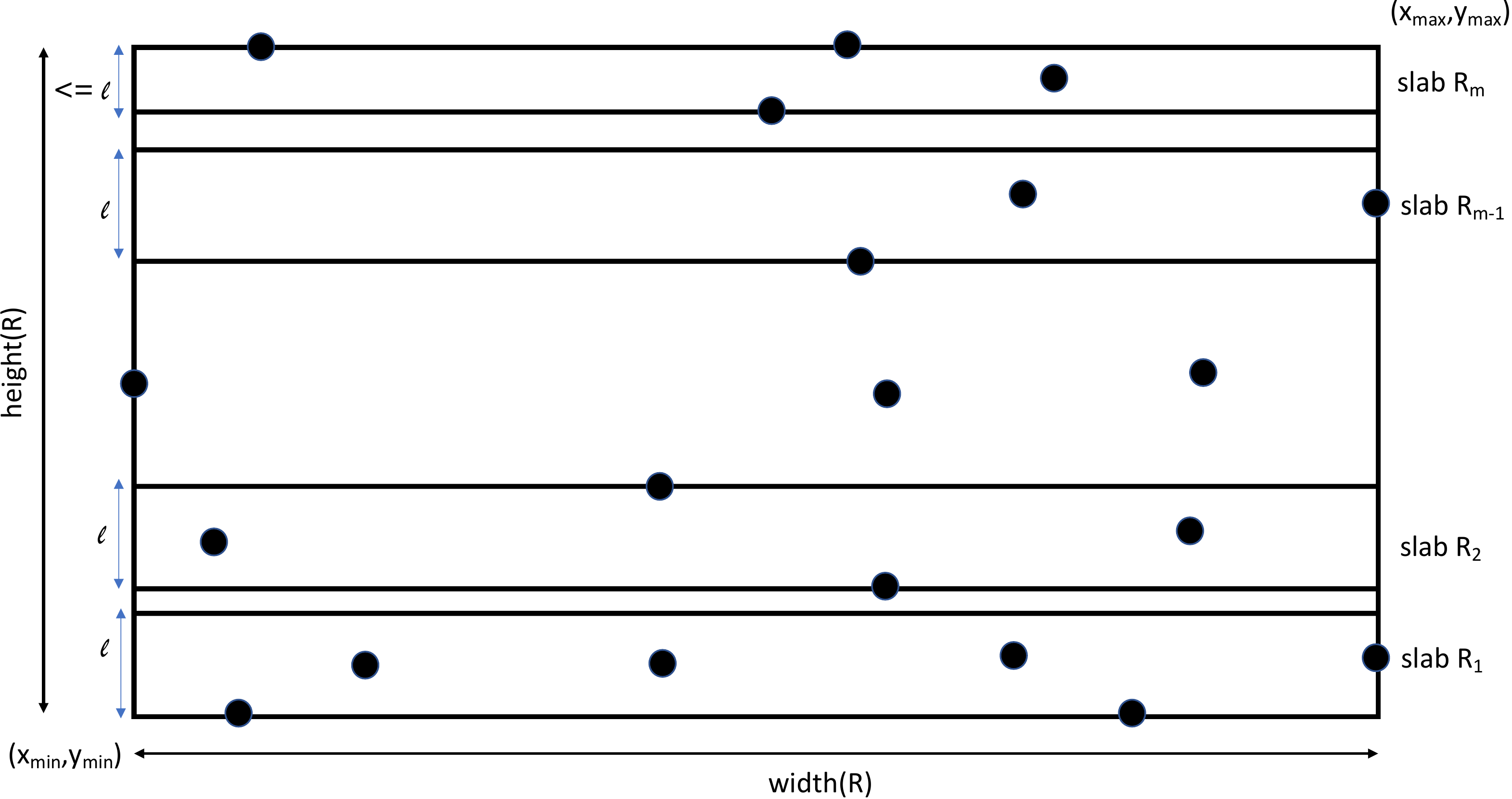}
\end{center}
\vspace{-6mm}
\caption{Selection of axis-aligned smallest enclosing rectangle $R$ covering all $N$ processes in $\cP$ and division of $R$ into axis-aligned slabs $R_i$ of height $\ell$ and width $width(R)$. The slabs are selected such that the bottom side of each slab $R_i$ has at least one process positioned on it.} 
\label{fig:rectangle}
\vspace{-4mm}
\end{figure}

Cover 
rectangle $R$ by a set $\cR$ of $m$ {\em slabs} $\cR=\{R_1,R_2,\ldots,R_m\}$.  The height of each slab $R_i$ is  $\ell$, except for the last slab $R_m$ whose height may be less than  $\ell$. The width of each slab is $width(R)$. That is this width is the same is the width of $R$. 

This slab-covering is done as follows. Let  $y_1=y_{min}+\ell$.
The area of $R$ between two horizontal lines passing through $y_{min}$ to $y_1$ is the first slab $R_1$. 
Now consider only the processes in $R$ that are not covered by $R_1$. Denote that process set by $\cP'$.   
Consider the bottom-most process in $\cP'$, i.e., process $p_{min'}=(x_{min'},y_{min'})\in \cP'$. We have that $y_{min'}>y_{min}+\ell$.
Draw two horizontal lines passing through $y_{min'}$ and $y_{min'}+\ell$. The area of $R$ between these lines is in slab $R_2$. Continue this way until all the points in $\cP$ are covered by a slab. In the last slab $R_m$, it may be the case that its height $height(R_m)<\ell$. 

\begin{figure}[!t]
\vspace{-0mm}
\begin{center}
\includegraphics[height=.6in]{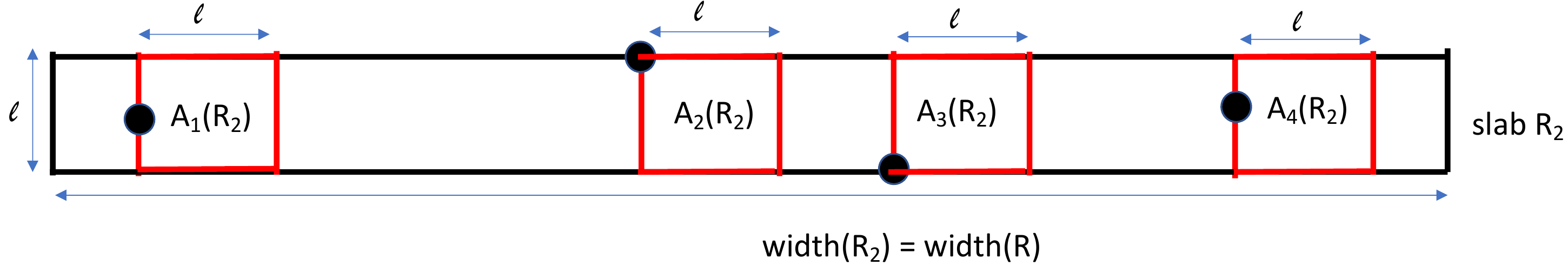}
\end{center}
\vspace{-7mm}
\caption{Selection of axis-aligned areas $A_j(R_2)$ (shown in red) to cover the processes in the slab $R_2$ of Figure~\ref{fig:rectangle}. The left side of each area $A_j(R_2)$ has at least one process positioned on it.}
\label{fig:slab}
\vspace{-5mm}
\end{figure}

So far, we covered $R$ by a set of $m$ slabs $\cR=\{R_1,\ldots,R_m\}$. We now we cover each such slab by axis-aligned square areas $A = \ell \times \ell$. See Figure~\ref{fig:slab} for illustration.
This square-covering is done as follows. Let $R_i$ be a slab to cover.
Put area $A$ on $R_i$ so that the top left corner of $A$ overlaps with the top left corner of slab $R_i$. Slide $A$ horizontally to the right so that there is a process in $R_i$ positioned on the left vertical line of $A$. Fix that area $A$ as one cover square and name it $A_1({R_i})$. Now consider only the points in $R_i$ not covered by  $A_1(R_i)$. It is immediate that those points are to the right of $A_1(R_i)$. Place $A$ on those points so that there is a point in $R_i$ positioned on the left side of $A$. Thus, there is no point of $R_i$ to the left of this second $A$ that is not covered by $A_1({R_i})$). Fix this as the second cover square and name it $A_2({R_i})$. Continue in this manner to cover all the points in $R_i$. Repeat this process for every slab of $R$.

\begin{lemma}
\label{lemma:slab-overlap}
Consider any two slabs $R_i,R_j\in \cR$ produced by GSQUARE.  $R_i$ and $R_j$ do not overlap, i.e., if some process $p\in R_i$, then $p\notin R_j$. 
\end{lemma}

\onlyLong{
\begin{proof}
It is sufficient to prove this lemma for adjacent slabs. Suppose slabs $R_i$ and $R_j$ are adjacent, i.e., $j=i+1$. According to the operation of \emph{GSQUARE}, after the location of $R_i$ is selected, only processes that are not covered by the slabs so far are considered for the selection of $R_j$. The first such process lies above the top (horizontal) side of $R_i$. Hence, there is a gap between the top side of $R_i$ and the bottom side of $R_j$.\qed
\end{proof}
}

\begin{lemma}
\label{lemma:square-overlap}
 Consider any two square areas $A_j({R_i})$ and $A_k(R_i)$ selected by GSQUARE in slab $R_i\in \cR$.  $A_j({R_i})$ and $A_k({R_i})$ do not overlap, i.e., if some process $p\in A_j({R_i})$, then $p\notin A_k({R_i})$.
\end{lemma}

\onlyLong{
\begin{proof}
It is sufficient to prove the lemma for adjacent squares. Suppose $A_j({R_i})$ and $A_k({R_i})$ are adjacent, i.e., $k=j+1$. Consider the operation of \emph{GSQUARE} in slab $R_i$ covered by $A_j({R_i})$ and $A_k({R_i})$. Area $A_k({R_i})$ only covers the processes that are not covered by $A_j({R_i})$ and, therefore, to the right of the right side of $A_j({R_i})$. As the left side of $A_k({R_i})$ is placed on the first such process, there is a non-empty gap between the two squares: $A_j({R_i})$ and $A_k({R_i})$.\qed  
\end{proof}
}

\begin{lemma}
\label{lemma:optimal-slab}
Consider slab $R_i\in \cR$. Let $k({R_i})$ be the number of squares $A_j({R_i})$ to cover all the processes in $R_i$ using GSQUARE. There is no algorithm that can cover the processes in $R_i$ with $k'(R_i)$ number of  squares $A_j({R_i})$ such that $k'(R_i)<k({R_i})$. 
\end{lemma}

\onlyLong{
\begin{proof}
Notice that slab $R_i$ has height $height(R_i)=\ell$ which is the same as the sides of (axis-aligned) squares $A_j(R_i)$ used to cover $R_i$.

\emph{GSQUARE} operates such that it places a square $A$ so that some process $p$ lies on the left side of this square. Consider a sequence of such processes: $\sigma \equiv \langle p_1 \cdots p_u, p_{u+1} \cdots p_j \rangle$.   Consider any pair of subsequent processes $p_u$ and $p_{u+1}$ in $\sigma$ with respective coordinates $(x_u, y_u)$
and $(x_{u+1}, y_{u+1})$. \emph{GSQUARE} covers them with
non-overlapping squares with side $\ell$. Therefore, $x_u + \ell < x_{u+1}$. That is, the distance between consequent processes in $\sigma$ is greater than $\ell$. Hence, any such pair of processes may not be covered by a 
single square. Since the number of squares placed by \emph{GSQUARE} in slab $R_i$ is $k$, the number of processes in $\sigma$ is also $k$. Any algorithm that covers these processes with axis-aligned squares requires at least $k$ squares. \qed
\end{proof}
}

Let $k_{opt}(\cR)$ be the number of axis-aligned square areas $A=\ell\times \ell$ to cover all $N$ processes in $R$ in the optimal cover algorithm. 
We now show that $k_{greedy}(\cR)\leq 2\cdot k_{opt}(\cR)$, i.e., \emph{GSQUARE}  provides 2-approximation. We divide the slabs in the set $\cR$ into two sets $\cR_{odd}$ and $\cR_{even}$. 
\onlyShort{For    $1\leq i\leq m$, let $\cR_{odd}:=\{R_i, i\mod 2\neq 0\} \text{~and~} \cR_{even}:=\{R_i, i\mod 2 = 0\}. $}
\onlyLong{For    $1\leq i\leq m$, let $$\cR_{odd}:=\{R_i, i\mod 2\neq 0\} \text{~and~} \cR_{even}:=\{R_i, i\mod 2 = 0\}. $$}

\begin{lemma}
\label{lemma:odd-even}
Let $k(\cR_{odd})$ and $k(\cR_{even})$ be the total number of (axis-aligned) square areas $A=\ell\times \ell$ to cover the processes in the sets $\cR_{odd}$ and $\cR_{even}$, respectively. Let $k_{opt}(\cR)$ be the optimal number of axis-aligned squares $A=\ell\times \ell$ to cover all the processes in $\cR$.  
$k_{opt}(\cR)\geq \max\{k(\cR_{odd}),k(\cR_{even})\}.$
\end{lemma}

\onlyLong{
\begin{proof}
Consider two slabs $R_i$ and $R_{i+2}$ for $i\geq 1$. Consider a square $A_j(R_i)$ placed by \emph{GSQUARE}. Consider also two processes $p\in R_i$ and $p'\in R_{i+2}$, respectively.  The distance between $p$ and $p'$ is $d(p,p')>\ell$. 
Therefore, if $A_j(R_i)$  covers $p$, then it cannot cover $p'\in R_{i+2}$.  Therefore, no algorithm can produce the optimal number of squares $k_{opt}(\cR)$ less than the maximum between $k(\cR_{odd})$ and $k(\cR_{even})$. 
\qed
\end{proof}
}

\begin{lemma}
\label{lemma:square-approximation}
$k_{greedy}(\cR)\leq 2\cdot k_{opt}(\cR)$.
\end{lemma}

\onlyLong{
\begin{proof}
From Lemma \ref{lemma:optimal-slab}, we obtain that \emph{GSQUARE} is optimal for each slab $R_i$. From Lemma \ref{lemma:odd-even}, we get that for any algorithm $k_{opt}(\cR)\geq \max\{k(\cR_{odd}),k(\cR_{even})\}.$
Moreover, the \emph{GSQUARE} produces the total number of squares $k_{greedy}(\cR)=k(\cR_{odd})+k(\cR_{even}).$
Comparing $k_{greedy}(\cR)$ with $k_{opt}(\cR)$, we get
$$\frac{k_{greedy}(\cR)}{k_{opt}(\cR)}\leq \frac{k(\cR_{odd})+k(\cR_{even})}{\max\{k(\cR_{odd}),k(\cR_{even})\}}\leq \frac{2\cdot \max\{k(\cR_{odd}),k(\cR_{even})\}}{\max\{k(\cR_{odd}),k(\cR_{even})\}} \leq 2. ~~~~~ \qed$$
\end{proof}
}

\vspace{1mm}
\noindent{\bf Covering by Circles.}
Let us formulate the covering by identical circles $C$ of diameter $\ell$, which we denote  CIRCLE-COVER. Let $\cA$ be the set of circles $C$.
We say that $\cA$ completely covers all the processes if every process $p_i\in \cP$ is covered by at least one of the circles in $\cA$. The following result can be established similar to SQUARE-COVER. 

\begin{theorem}
CIRCLE-COVER is NP-Complete.
\end{theorem}

\noindent{\bf A Greedy Circle Cover Algorithm.}
We call this algorithm \emph{GCIRCLE}. Pick the square cover set $\cA$ produced in Section \ref{subsection:greedy}. The processes covered by any square $A\in \cA$ can be completely covered by 4 circles $C$ of diameter $\ell$: Find the midpoints of the 4 sides of the square and draw the circles $C$ of diameter $\ell$ with their centers on those midpoints. 

\onlyShort{\vspace{-2mm}}
\begin{lemma}
\label{lemma:circle-approximation}
Let $k^C_{greedy}(\cR)$ be the number of circles $C$ of diameter $\ell$ needed to cover all the processes in $\cP$ by algorithm \emph{GCIRCLE}. $k^C_{greedy}(\cR)\leq 8\cdot k^C_{opt}(\cR)$, where $k^C_{opt}(\cR)$ is the optimal number of circles $C$ in any algorithm.
\end{lemma}

\onlyLong{
\begin{proof}
We first show that 
$k^C_{opt}(\cR)\geq \max\{k^S(\cR_{odd}),k(\cR^S_{even})\},$
where $k^S(\cR_{odd})$ and $k^S(\cR_{even})$, respectively, are the number of squares $A=\ell\times \ell$ to cover the slabs in $\cR_{odd}$ and   $\cR_{even}$. Consider any square cover $A_j(R_i)$ of any slab $R_i$. A circle $C$ of diameter $\ell$ can cover at most the processes in $A_j(R_i)$ but not in any other square $A_l(R_i)$. This is because the perimeter of $C$ needs to pass through the left side of $A_j(R_i)$ (since there is a process positioned on that line in $A_j(R_i)$) and with diameter $\ell$, the perimeter of $C$ can touch at most the right side of $A_j(R_i)$.  

We now prove the upper bound. Since one square area $A=\ell\times \ell$ is now covered using at most 4 circles $C$ of diameter $\ell$, \emph{GCIRCLE} produces the total number of circles $k^C_{greedy}(\cR)=4\cdot (k^S(\cR_{odd})+k^S(\cR_{even})).$

Comparing $k^C_{greedy}(\cR)$ with $k^C_{opt}(\cR)$ as in  Lemma \ref{lemma:square-approximation}, we have that 
$$\frac{k^C_{greedy}(\cR)}{k^C_{opt}(\cR)}\leq \frac{4\cdot (k^S(\cR_{odd})+k^S(\cR_{even}))}{\max\{k^S(\cR_{odd}),k^S(\cR_{even})\}}\leq \frac{8\cdot \max\{k^S(\cR_{odd}),k^S(\cR_{even})\}}{\max\{k^S(\cR_{odd}),k^S(\cR_{even})\}} 
\leq 
8.~\qed$$
\end{proof}
}

\noindent{\bf Overlapping Fault Area.}
The adversary may place the fault area $F$ in any location in the plane. This means that $F$ may not necessarily be axis-aligned. Algorithms \emph{GSQUARE} and \emph{GCIRCLE} produce a cover set $\cA$ of axis-aligned squares and circles, respectively. Therefore, the algorithm we present in the next section needs to know how many areas in $\cA$ that $F$ overlaps. We now compute the bound on this number. 
The bound considers both square and circle areas $A$ under various size combinations of fault and non-fault areas.  
The lemma below is for each $A\in \cA$ and $F$ being either squares of side $\ell$ or circles of diameter $\ell$. 

\begin{lemma}
\label{lemma:overlap-A}
For the processes in $\cP$, consider the cover set $\cA$ consisting of the axis-aligned square areas $A=\ell\times \ell$. Place a relocatable square area $F=\ell\times \ell$ in any orientation (not necessarily axis-aligned). $F$ overlaps no more than 7 squares $A$.  If the cover set consists of circles $C\in \cA$  of diameter $\ell$ and $F$ is a circle of diameter $\ell$, then $F$ overlaps no more than $28$ circles $C$.   
\end{lemma}

\onlyLong{

\begin{proof}
Suppose $F$ is axis-aligned.  $F$ may overlap at most two squares $A$ horizontally. Indeed, the total width covered by two squares in $\cA$ is $>2\ell$ since the squares do not overlap. Meanwhile, the total width of $F$ is $\ell$. 
Similarly, $F$ may overlap at most two squares vertically. Combining possible horizontal and vertical overlaps, we obtain that $F$ may overlap at most 4 distinct axis-aligned areas $A$. See Figure~\ref{fig:overlap-aaligned} for illustration.

\begin{figure}[!t]
\vspace{-0mm}
\begin{center}
\includegraphics[height=1.0in]{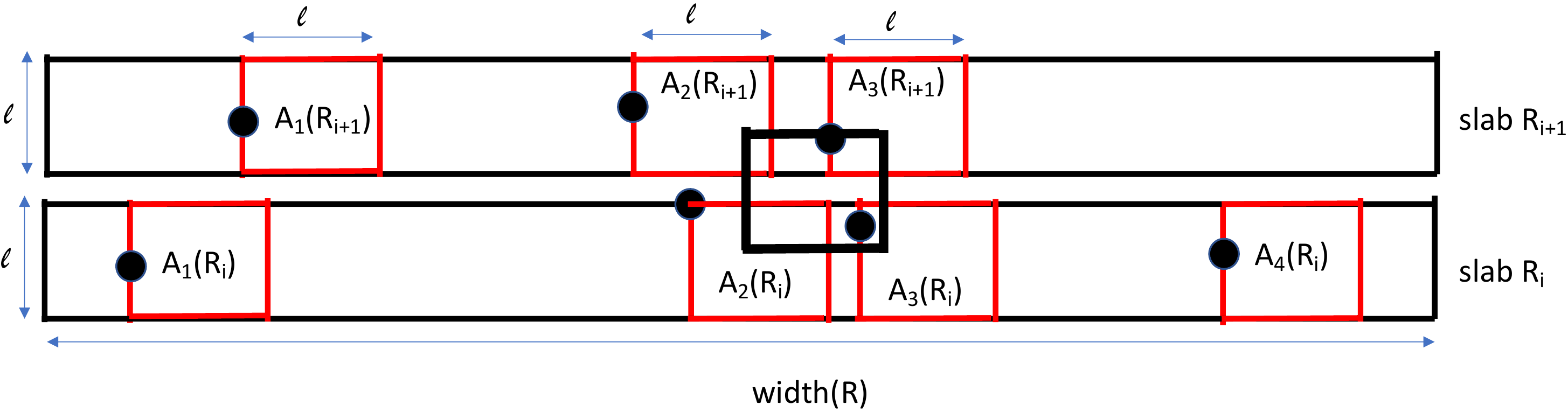}
\end{center}
\vspace{-7mm}
\caption{The maximum overlap of an axis-aligned fault area $F$ with the identical axis-aligned cover squares $A$ of same size.}
\label{fig:overlap-aaligned}
\vspace{-0mm}
\end{figure}

\begin{figure}[!t]
\vspace{-2mm}
\begin{center}
\includegraphics[height=1.15in]{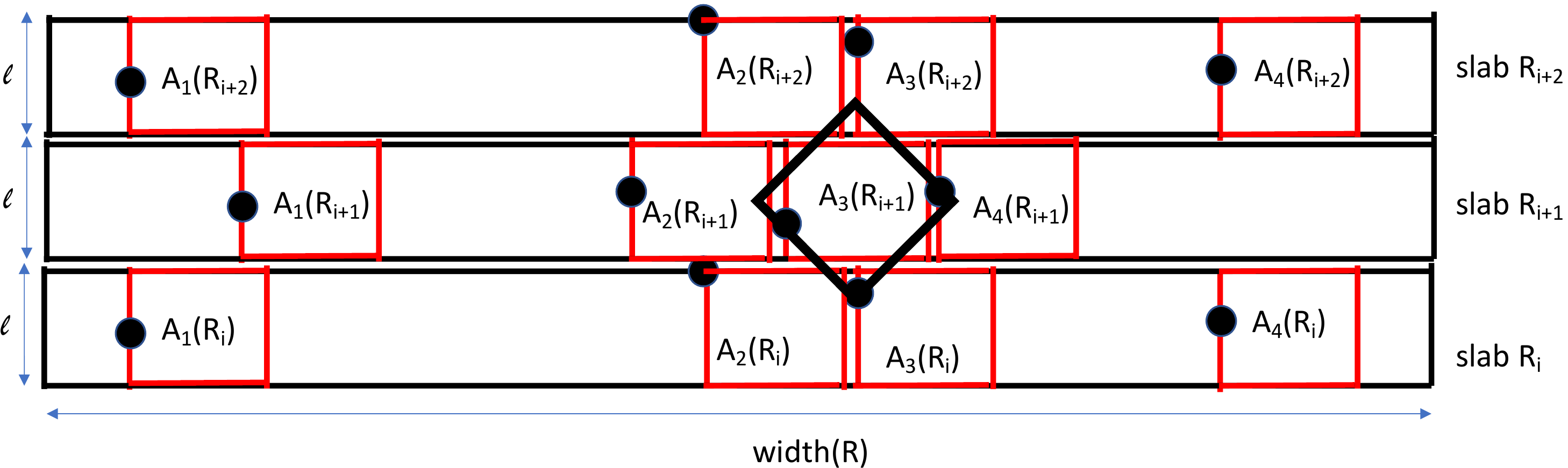}
\end{center}
\vspace{-7mm}
\caption{The maximum overlap of a non-axis-aligned fault area $F$ with the identical axis-aligned cover squares $A$ of the same size.}
\label{fig:overlap-nonaligned}
\vspace{-5mm}
\end{figure}

Consider now that $F$ is not axis-aligned. 
$F$ can span at most $\sqrt{2}\ell$ horizontally and $\sqrt{2}\ell$ vertically. Therefore, horizontally, $F$ can overlap at most three areas $A$.  Vertically, $F$ can overlap three areas as well. However, not all three areas on the top and bottom rows can be overlapped at once. Specifically, not axis-aligned $F$ can only overlap  2 squares in the top row and 2 in the bottom row. Therefore, in total, $F$ may overlap at most 7 distinct axis-aligned areas. Figure~\ref{fig:overlap-nonaligned} provides an illustration. 

For the case of circular $F$, one square area $A$ can be completely covered by 4 circles $C$. Furthermore, square $F$ of size $\ell$ overlaps at most 7 square areas $A$ of side $\ell$. Moreover, the circular $F$ of diameter $\ell$ can be inscribed in a square of side $\ell$. Therefore, a circular $F$ cannot overlap more than 7 squares, and hence  the circular $F$ may overlap in total at most $7\times 4=28$ circles $C$.    
\qed  
\end{proof}
}

The first lemma below is for each $A$ being an axis-aligned square of side $\ell$ or a circle of diameter $\ell$ while $F$ being either a square of side $\ell/\sqrt{2}$ or a circle of diameter $\ell/\sqrt{2}$.  The second lemma below considers circular fault area $F$ of diameter $\sqrt{2}\ell$.

\begin{lemma}
\label{lemma:overlap-Aroot2}
For the processes in $\cP$, consider the cover set $\cA$ consisting of the axis-aligned squares  $A=\ell\times \ell$. Place a relocatable square area $F=\ell/\sqrt{2}\times \ell/\sqrt{2}$ in any orientation (not necessarily axis-aligned). $F$ overlaps no more than  4 squares  $A$.  If the cover set consists of circles $C\in \cA$ of diameter $\ell$ each, and $F$ is a circle of diameter $\ell/\sqrt{2}$, then $F$ overlaps no more than 16 circles $C$.
\end{lemma}

\onlyLong{
\begin{proof}
$F$ can extend, horizontally and vertically, at most $\sqrt{2}\cdot \ell/\sqrt{2}=\ell.$ Therefore, $F$ can overlap no more than two squares $A$ horizontally and two squares $A$ vertically. 

For the case of circular $F$ of diameter $\ell/\sqrt{2}$, it can be inscribed in a square of side $\ell/\sqrt{2}$. 
This square can overlap no more than $4$ squares of $\ell\times\ell$. Each such square can be covered by at most $4$ circles of diameter $\ell$. Therefore, the total number of circles to overlap the circular fault area $F$ is $4\times4 = 16$.
\qed
\end{proof}
}


\begin{lemma}
\label{lemma:overlap-circle}

For the processes in $\cP$, consider the cover set $\cA$ consisting of the axis-aligned squares areas $A=\ell\times \ell$. Place a relocatable circular fault area $F$ of diameter $\sqrt{2}\ell$. $F$ overlaps no more than 8 squares $A$. If $\cA$ consists of circles $C$ of diameter $\ell$, then circular $F$ of diameter  $\sqrt{2}\ell$ overlaps no more than 32 circles $C$. 
\end{lemma}


\onlyLong{
\begin{figure}[!t]
\vspace{-0mm}
\begin{center}
\includegraphics[height=1.12in]{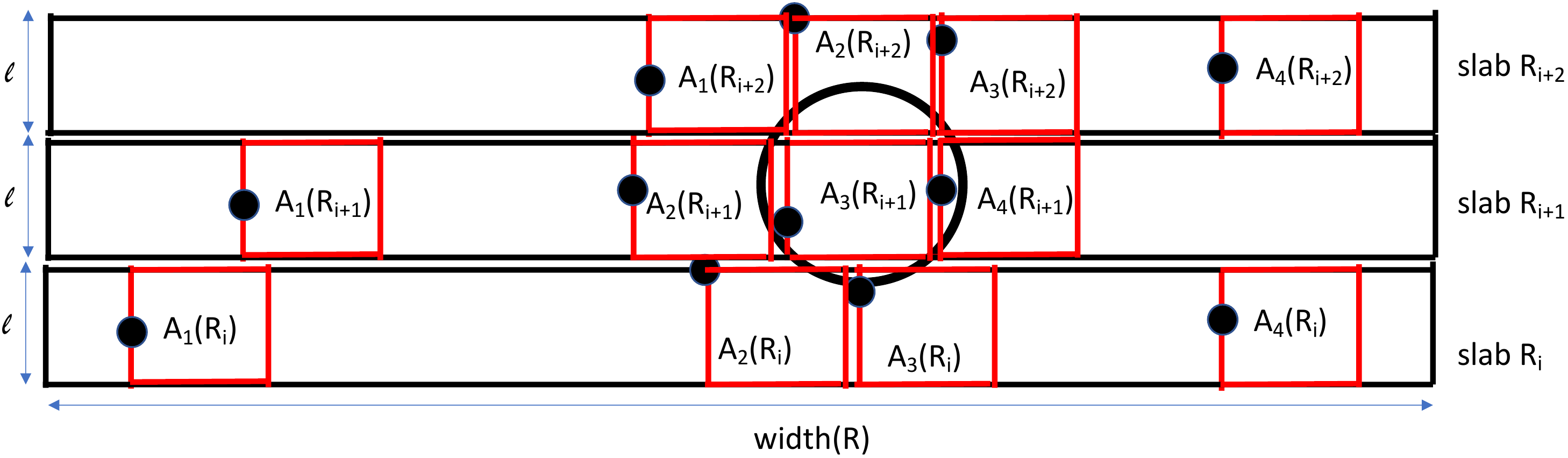}
\end{center}
\vspace{-6mm}
\caption{The maximum overlap of a circular fault area $F$ of diameter $\sqrt{2}\ell$ with axis-aligned cover squares $A$ of side $\ell$.}
\label{fig:overlap-circle}
\vspace{-2mm}
\end{figure}

\begin{proof}
Since $F$ is a circle of diameter $\sqrt{2}\ell$, $F$ can span horizontally and vertically at most $\sqrt{2}\cdot \ell$. 
Arguing similarly as in Lemma \ref{lemma:overlap-A}, $F$ can overlap either at most 3 squares $A$ in top row or 3 on the bottom row. Interestingly, if $F$ overlaps 3 squares in the top row, it can only overlap at most 2 in the bottom row and vice-versa. Therefore, in total, $F$ overlaps at most 8 distinct squares of side $\ell$. Figure~\ref{fig:overlap-circle} provides an illustration. 

Since one square of side $\ell$ can be completely covered using 4 circles of diameter $\ell$, $F$ of diameter $\sqrt{2}\ell$ can cover at most $8\times 4=32$ circles $C$ of diameter $\ell$. \qed
\end{proof}
}

\onlyLong{
\begin{algorithm}[!t]
{\small
{\bf Setting:} A set $\cP$ of $N$ processes positioned at distinct planar coordinates. Each process can communicate with all other processes and knows the coordinates of all other processes. The processes covered by the fault area $F$ at unknown location are Byzantine. There are $M\geq 1$ of identical fault areas $F$ and processes know $M$.\\
{\bf Input:} Each process has initial value either 0 or 1.\\
{\bf Output:} Each correct process outputs decision subject to geoconsensus\\
{\bf \em Procedure for process $p_k$}\\
// leaders selection \\
Compute the set $\cA$ of covers $A_j(R_i)$;\\
// {\bf For each} cover $A_j(R_i)\in \cA$ {\bf do}\\
 
$\cP_{min}\leftarrow$ a set of processes with minimum $y$-coordinate among covered by $A_j(R_i)$;\\
\If{$|\cP_{min}|=1$}
{
$l_j(A_j(R_i))\leftarrow$ the only process in $\cP_{min}$;\\
}
\Else
{
$l_j(A_j(R_i))\leftarrow$ the process in  $\cP_{min}$ with minimum $x$-coordinate;\\
}

Let $P_L$ be the set of leaders, one for each $A_j(R_i)\in \cA$; \\
// consensus \\
\eIf{$p_k \in \cP_L$}
{run \emph{PSL} algorithm, achieve decision $v$, broadcast $v$, output $v$}
{
wait for messages with identical decision $v$ from at least $2M + 1$ processes from $\cP_L$, output $v$ 
}
\caption{\emph{GENERIC} geoconsensus algorithm.}
\label{algorithm:consensus}
}
\end{algorithm}	
}

\vspace{-2mm}
\section{The GENERIC Geoconsensus Algorithm}
\vspace{-2mm}
\label{section:generic}
We now describe an algorithm solving 
geoconsensus we call {\em GENERIC} for a set $\cP$ of $N$ processes on the plane. 
Each process $p_k$
knows the coordinates of all other  processes and can communicate with all of them.
Each process $p_k$ knows the shape (circle, square, etc.) and size (diameter, side, etc.) of the fault area $F$. There are $M\geq 1$ fault areas, i.e., $|\mathcal{F}|=M$ and $p_k$ knows $M$. The processes do not know the orientation and location of each fault area $F$. Fault area $F$ is controlled by an adversary and all processes covered by that area $F$ are Byzantine.  Each process $p_k$ is given an initial value either 0 or 1. The output of each process has to comply with the three properties of geoconsensus. 

The pseudocode is given in Algorithm \ref{algorithm:consensus}\onlyShort{~(in Appendix)}.
{\em GENERIC} operates as follows.
Each process $p_k$ computes a set $\cA$ of covers $A_j(R_i)$ that are of same size as $F$. Then $p_k$ determines the leader $l_j(A_j(R_i))$ in each cover $A_j(R_i)$. The process in $A_j(R_i)$ with smallest $y$-coordinate is selected as a leader. If there exist two processes with the same smallest $y$-coordinate, then the process with the smaller $x$-coordinate between them is picked. 
If $p_k$ is selected leader, it participates in running \emph{PSL}  \cite{Pease80}. The leaders run \emph{PSL} then broadcast the achieved decision. The non-leader processes adopt it. 

\vspace{2mm}
\noindent{\bf Analysis of \emph{GENERIC}.}
We now study the correctness and fault-tolerance guarantees of \emph{GENERIC}. In all theorems of this section, \emph{GENERIC} achieves the solution in $M+2$ communication rounds. The proof for this claim is similar to that for \emph{BASIC} in Theorem~\ref{theorem:basic}.

Let the fault area $F=\ell\times \ell$ be a, not necessarily axis-aligned, square. 

\begin{theorem}
\label{theorem:single-fault-main}
Given a set $\cP$ of $N$ processes and one square are $F$ are positioned at an unknown location such that any process of $\cP$ covered by $F$ is Byzantine. Algorithm GENERIC solves geoconsensus 
with the following guarantees:
\begin{itemize}
    \item If $F=\ell\times \ell$ and not axis-aligned and  $A=\ell\times \ell$, $f\leq N-15$ faulty processes can be tolerated given that $|\cA|\geq 22$. 
    \item If $F=\ell\times \ell$ and  axis-aligned and $A=\ell\times \ell$, $f\leq N-9$ faulty processes can be tolerated given that  $|\cA|\geq 13$.
    \item If $F=\ell/\sqrt{2}\times \ell/\sqrt{2}$ but $A=\ell\times \ell$, then even if $F$ is not axis aligned, $f\leq N-9$ faulty processes can be tolerated given that $|\cA|\geq 13$.
\end{itemize} 
\end{theorem}
\begin{proof}
We start by proving the first case. From Lemma \ref{lemma:overlap-A}, we obtain that a square fault 
area $F=\ell\times \ell$, regardless of orientation and location, can overlap at most $n(F)=7$ axis-aligned squares $A=\ell\times \ell$. 
When $|\cA|\geq 22$, we have at least $\cA-n(F)\geq 15$ axis-aligned squares containing only correct processes.
Since \emph{GENERIC} reaches consensus using only the values of the leader processes in each area $A$, if we have $|\cA|\geq 22$ areas, it is guaranteed that $\geq 2\cdot |\cA|/3+1\geq 2\cdot n(F)+1$ leader processes are correct (with $n(F)=7$) and they can reach consensus using \emph{PSL} algorithm.   
Regarding the number of faulty process that can be tolerated, the fault area $F$ can cover $f\leq N-15$ processes but still algorithm \emph{GSQUARE} produces total $|\cA|=22$ areas. All these $f\leq N-15$ faulty processes can be tolerated.   

Let us address the second case. An axis-aligned square $F$ can overlap at most $n(F)=4$ axis-aligned squares $A$. Therefore, when $|\cA|\geq 13$, we have that $|\cA|-9\geq 2\cdot n(F)+1$ leader processes are correct and they can reach consensus. In this case, $f\leq N-9$ processes can be covered by $F$ and still they all can be tolerated. 

Let us now address the third case, when $F=\ell/\sqrt{2}\times \ell/\sqrt{2}$ but $A=\ell\times \ell$. Regardless of its orientation, $F$ can overlap at most $n(F)=4$ squares $A$. Therefore, $|\cA|\geq 13$ is sufficient for consensus and total $f\leq N-9$ processes can be tolerated.  
\qed 
\end{proof}

For the multiple fault areas $F$ with $|\mathcal{F}|=M$, Theorem \ref{theorem:single-fault-main} extends as follows. 

\begin{theorem}
\label{theorem:multiple-fault-main}
Given a set $\cP$ of $N$ processes and a set of  $M\geq 1$ of square areas $F$ positioned at unknown locations such that any process of $\cP$ covered by any $F$ may be Byzantine.
Algorithm GENERIC solves geoconsensus with the following guarantees:
\begin{itemize}
    \item If each $F=\ell\times \ell$ and not axis-aligned and $A=\ell\times \ell$, $f\leq N-15M$ faulty processes can be tolerated given that $|\cA|\geq 22M$. 
    \item If each $F=\ell\times \ell$ and  axis-aligned and $A=\ell\times \ell$, $f\leq N-9M$ faulty processes can be tolerated given that $|\cA|\geq 13 M$. 
    \item If each $F=\ell/\sqrt{2}\times \ell/\sqrt{2}$ but $A=\ell\times \ell$, then even if $F$ is not axis-aligned, $f\leq N-9M$ faulty processes can be tolerated given that $|\cA|\geq 13M$.
\end{itemize} 
\end{theorem}

\onlyLong{
\begin{proof}
The proof for the case of $M=1$ extends to the case of $M>1$ as follows. Theorem \ref{theorem:single-fault-main} gives the bounds $f\leq N-\gamma$ and $|\cA|\geq \delta$ for one fault area for some positive integers $\gamma,\delta$. For $M$ fault areas, $M$ separate $|\cA|$ sets are needed, with each set tolerating a single fault area $F$. Therefore, the bounds of Theorem~\ref{theorem:single-fault-main} extend to multiple fault areas with a factor of $M$, i.e., \emph{GENERIC} needs $M\cdot \delta$ covers and $f\leq N-M \cdot \gamma$ faulty processes can be tolerated.   
Using the appropriate numbers from Theorem~\ref{theorem:single-fault-main} provides the claimed bounds.
\qed
\end{proof}
}

We have the following theorem for the case of circular fault set $\mathcal{F}$,  $|\mathcal{F}|=M\geq 1$.

\begin{theorem}
\label{theorem:multiple-fault-main-circle}
Given a set $\cP$ of $N$ processes  and a set of $M\geq 1$  circles $F$ positioned at unknown locations such that any process of $\cP$ covered by $F$ may be Byzantine. Algorithm GENERIC solves geoconsensus with the following guarantees:
\begin{itemize}
    \item If each $F$ and $A$ are circles of diameter $\ell$, $f\leq N-57M$ faulty processes can be tolerated given that $|\cA|\geq 85M$. 
    \item If each $F$ is a circle of diameter $\sqrt{2}\ell$ and $A$ is a circle of diameter $\ell$, $f\leq N-65M$ faulty processes can be tolerated given that $|\cA|\geq 97 M$. 
    \item If each $F$ is a circle of diameter $\ell/\sqrt{2}$ and $A$ is a circle of diameter $\ell$, $f\leq N-33M$ faulty processes can be tolerated given that $|\cA|\geq 49 M$. 
\end{itemize} 
\end{theorem}

\onlyLong{
\begin{proof}
For the first case, we have that $n(F)=28$, when cover set $\cA$ is of circles of diameter $\ell$ and the fault area $F$ is also a circle of diameter $\ell$. Therefore, when $|\cA|\geq 85M$, we have that at least $|\cA|-n(F)\geq 57M$ circles containing only correct processes. Since Algorithm \ref{algorithm:consensus} reaches consensus using only the values of the leader processes in each area $A$, when we have $|\cA|\geq 85M$, it is guaranteed that $\geq 2\cdot |\cA|/3+1\geq 2\cdot n(F)M+1$ leader processes are correct and hence \emph{GENERIC} can reach consensus. The fault tolerance guarantee of $f\leq N-57M$ can be shown similarly to the proof of Theorem \ref{theorem:single-fault-main}.

For the second result, we have shown that $n(F)=32$. Therefore, we need $|\cA|\geq 3\cdot n(F)+1\geq 97$ for one faulty circle $F$ of diameter $\sqrt{2}\ell$. For $M$ faulty circles, we need $|\cA|\geq 97M$. Therefore, the fault tolerance bound is $f\leq N-(2\cdot n(F)M+1)=N-65M$. 

For the third result, we have shown that $n(F)=16$ for a single faulty circle of diameter $\ell/\sqrt{2}$. Therefore, we need $|\cA|\geq 49M$ and $f\leq N-33M$. 
\qed
\end{proof}
}

\section{Extensions to Higher Dimensions}
\label{section:extension}
\vspace{-2mm}
Our approach can be extended to solve geoconsensus in $d$-dimensions, $d\geq 3$.
\emph{BASIC} extends as is, whereas
\emph{GENERIC} runs without modifications in higher dimensions so long as we determine (i) the cover set $\cA$ of appropriate dimension and (ii) the overlap bound -- the maximum number of $d$-dimensional covers $A$ that the fault area $F$ may overlap. 
The bound on $f$ then depends on $M$ and the cover set size $|\cA|$.  
In what follows, we discuss 3-dimensional space.  The still higher dimensions can be studied similarly.

When $d=3$, the objective is to cover the embedded processes of $\cP$ by cubes of size $\ell\times\ell\times\ell$ or spheres of diameter $\ell$. It can be shown that the greedy cube (sphere) cover algorithm, let us call it \emph{GCUBE} (\emph{GSPHERE}), provides $2^{d-1}=4$ (16)  approximation of the optimal cover. The idea is to appropriately extend the 2-dimensional slab-based division  and axis-aligned square-based covers discussed in Section~\ref{subsection:greedy} to $3$-dimensions with rectangular cuboids and cube-based covers. \onlyShort{The detailed discussion is in Appendix.}

\onlyLong{
Suppose the coordinates of process $p_i\in \cP$ are $(x_i,y_i,z_i)$. 
\emph{GCUBE} operates as follows. It first finds $x_{min},y_{min},z_{min}$ as well as $x_{max},y_{max},z_{max}$. Then, a smallest axis-aligned (w.r.t. $x$-axis) cuboid, i.e. rectangular parallelepiped,  $R$ with the left-bottom-near corner $(x_{min},y_{min},z_{min})$ and the right-top-far corner at $(x_{max},y_{max},z_{max})$ is constructed such that $R$ covers all $N$ processes in $\cP$. Assume that $z-axis$ is away from the viewer. The depth of $R$ is $depth(R)=z_{max}-z_{min}$; $width(R)$ and $height(R)$ are similar to \emph{GSQUARE}.    

\emph{GCUBE} now divides $R$ into a set $\cR$ of $m$ cuboids $\cR=\{R_1,\cdots,R_m\}$ such that $depth(R_i)=\ell$ but the $width(R_i)=width(R)$ and   $height(R_i)=height(R)$. 
Each $R_i$ is further divided into a set of $\cR_i$ of $n$ cuboids $\cR_i=\{R_{i1},\ldots,R_{in}\}$ such that each $R_{ij}$ has $width(R_{ij})=width(R)$ but $height(R_{ij})=\ell$ and $depth(R_{ij})=\ell$. 
Each cuboid $R_{ij}$ is similar to the slab $R_i$ shown in Figure \ref{fig:slab} but has depth $\ell$. 

It now remains to cover each axis-aligned cuboid $R_{ij}$ with cubic areas $A$ of side $\ell$. Area $A$ can be put on $R_{ij}$ so that the top left corner of $A$ overlaps with the top left corner of cuboid $R_{ij}$. Slide $A$ on the $x$-axis to the right so that there is a process covered by $R_{ij}$ positioned on the left vertical plane of $A$. Fix that area $A$ as one cover cube and name it $A_1(R_{ij})$. Now consider only the processes in $R_{ij}$ not covered by $A_1(R_{ij})$. Place another $A$ on those processes so that there is a point in $R_{ij}$ positioned on the left vertical plane of $A$ and there is no process on the left of $A$ that is not covered by $A_1(R_{ij})$. Let that $A$ be $A_2(R_{ij})$. Continue this way to cover all the processes in $R_{ij}$.  

Apply the procedure of covering $R_{ij}$ to all  $m\times n$ cuboids.
Lemma \ref{lemma:slab-overlap} can be extended to show that no two cuboids $R_{ij}, R_{kl}$ overlap. Lemma \ref{lemma:square-overlap} can be extended to show that no two cubic covers $A_o(R_{ij})$ and $A_p(R_{kl})$ overlap. For each cuboid $R_{ij}$, Lemma \ref{lemma:optimal-slab} can be extended to show that no other algorithm produces the number of cubes $k'(R_{ij})$ less than the number of cubes $k(R_{ij})$ produced by algorithm \emph{GCUBE}.

Since the cover for each square cuboid $R_{ij}$ is individually optimal, 
let $k_{opt}(\cR)$ be the number of axis-aligned cubes to cover all $N$  processes in $R$ in the optimal cover algorithm. We now show that $k_{greedy}(\cR)\leq 4\cdot k_{opt}(\cR)$, i.e., \emph{GCUBE} provides 4-approximation. 
We do this by combining two approximation bounds. The first is for the $m$ cuboids $R_i$, for which we show $2$-approximation. We then provide $2$-approximation for each cuboid $R_i$ which is now divided into $n$ cuboids $R_{ij}$. Combining these two approximations, we have, in total, a $4$-approximation.  

As in the 2-dimensional case, divide the $m$ cuboids in the set $\cR$ into two sets $\cR_{odd}$ snd $\cR_{even}$. Arguing as in Lemma \ref{lemma:optimal-slab}, we can show that $k_{opt}(\cR)\geq \max\{k(R_{odd}),k(R_{even})\}$ and $k_{greedy}(\cR)=k(R_{odd})+k(R_{even})$. 
Therefore, the ratio $k_{greedy}(\cR)/k_{opt}(\cR)\leq 2$ while dividing $R$ into $m$ cuboids.

Now consider any cuboid $R_i\in \cR_{odd}$ ($R_i\in \cR_{even}$ case is analogous). 
$R_i$ is divided into a set $\cR_i$ of $n$ cuboids $R_{ij}$. Divide $n$ cuboids in the set $\cR_i$ into two sets $\cR{i,odd}$ and $\cR{i,even}$ based on odd and even $j$. Therefore, it can be shown that, similarly to Lemma \ref{lemma:optimal-slab}, that  $k_{opt}(\cR_i)\geq \max\{k(R_{i,odd}),k(R_{i,even})\}$ and $k_{greedy}(\cR_i)=k(R_{i,odd})+k(R_{i,even})$.  Therefore, $k_{greedy}(\cR_i)/k_{opt}(\cR_i)\leq 2$. Combining the $2$-approximations each for the two steps, we have the overall $4$-approximation.

Let us now discuss the $16$-approximation for spheres of diameter $\ell$. One cube $A_l(R_{ij})$ of side $\ell$ can be completely covered by $4$ spheres of diameter $\ell$. Since, for cubes, \emph{GCUBE} is  $4$-approximation, we, therefore, obtain that \emph{GSPHERE} is a $16$-approximation. 
We omit this discussion but it can be shown that \emph{GSPHERE}, appropriately extended from \emph{GCIRCLE} into 3-dimensions, achieves $(2^{d-1}\cdot d^d)=4\cdot 27=108$ approximation.  

Now we need to find the overlap number $n(F)$.  Cube $A$ of side $\ell$ has diameter $D=\sqrt{3}\ell$. That means that a cubic fault area $F$ that has the same size as $A$ can overlap at most 3 cubes $A_l(R_{ij})$ in all 3 dimensions.  Therefore, $F$ can cover at most $3^3=27$  cubes $A_l(R_{ij})$.
For sphere $F$ of diameter $\ell$, since one cube $A_l(R_{ij})$ can be completely covered by $4$ spheres of diameter $\ell$ and $F$ can be inscribed inside $A_l(R_{ij})$, it overlaps the total $4\cdot 27=108$ spheres $A_l(R_{ij})$. For the the axis-aligned case of cubic fault area $F$, it can be shown that $n(F)=8$ cubes $A_l(R_{ij})$. This is because it can overlap with at most $4$ cubes $A_l(R_{ij})$ as Figure ~\ref{fig:overlap-aaligned} and, due to depth $\ell$, it can go up to two layers, totaling 8. $n(F)=32$ for sphere $F$ is immediate since each cube $A_l(R_{ij})$ is covered by $4$ spheres of diameter $\ell$, sphere of diameter $\ell$ can be inscribed inside a cube $A_l(R_{ij})$ of side $\ell$, and a faulty cube $F$ of side $\ell$ can overlap at most 8 axis-aligned cubes $A_l(R_{ij})$.
}

We summarize the results for cubic covers and cubic fault areas in Theorem \ref{theorem:multiple-fault-main-3d-cube}. 
\vspace{-1mm}
\begin{theorem}
\label{theorem:multiple-fault-main-3d-cube}
Given a set $\cP$ of $N$ processes embedded in 3-d space and a set of $M\geq 1$ of cubic areas $F$ at unknown locations, such that any process of $\cP$ covered by $P$ may be Byzantine. Algorithm GENERIC solves geoconsensus with the following guarantees:
\begin{itemize}
    \item If $F$ is cube of side $\ell$ and not axis-aligned and $A$ is also a cube of side $\ell$, $f\leq N-55M$ faulty processes can be tolerated given that the cover set $|\cA|\geq 82M$. 
    \item If $F$ is cube of side $\ell$ and axis-aligned and $A$ is also a cube of side $\ell$, $f\leq N-17M$ faulty processes can be tolerated given that $|\cA|\geq 25 M$.
\item If $F$ is a sphere of diameter $\ell$ and  $A$ is a sphere of diameter $\ell$, $f\leq N-217M$ faulty processes can be tolerated given that $|\cA|\geq 325 M$. 
\end{itemize}
\end{theorem}    

\section{Concluding Remarks} 
\label{section::end}
\vspace{-2mm}
Byzantine consensus is a relatively old, practically applicable and well-researched problem. It had been attracting extensive attention from researchers and engineers in distributed systems. In light of the recent development on location-based consensus protocols, such as G-PBFT \cite{LaoD0G20}, we have formally defined and studied the consensus problem of processes that are embedded in a $d$-dimensional plane, $d\geq 2$. We have explored both the possibility as well bounds for a solution to geoconsensus. Our results provide trade-offs on three parameters $N,M,$ and $f$, in constant to the trade-off between only two parameters $N$ and $f$ in the Byzantine consensus literature.  Our results also show the dependency of the tolerance guarantees on the shapes of the fault areas.  


For future work, it would be interesting to close or reduce the gap between the condition for impossibility and a solution (as discussed in Contributions). It would also be interesting to consider fault area $F$ shapes beyond circles and squares that we studied; to investigate process coverage by non-identical squares, circles or other shapes to see whether better bounds on the set $\cA$ and fault-tolerance guarantee $f$ can be obtained.  

\footnotesize
\bibliographystyle{splncs03}
\bibliography{references}
\onlyShort{
\normalsize
\newpage
\setcounter{page}{1}
\section*{Appendix}

\begin{table}[!t]
\vspace{-4mm}
{\footnotesize
\centering
\begin{tabular}{l|l}
\toprule
{\bf Symbol} & {\bf Description}  \\
\toprule
$N$; $\cP$; $(x_i,y_i)$ &  number of processes; $\{p_1,\ldots,p_N\}$; planar coordinates of process $p_i$\\
\hline
$F; D$; $\mathcal{F}$ & fault area; diameter of $F$; a set of fault areas $F$ with $|\mathcal{F}|=M$\\
\hline
$f$ & number of faulty processes\\
\hline
$\cP_D$ & processes in $\cP$ such that pairwise distance between them is more than $D$\\
\hline
$A$ (or $A_j(R_i)$); $\cA$ & cover area that is of same shape and size as $F$; a set of cover areas $A$\\
\hline
$n(F)$ & number of cover areas $A\in \cA$ that a fault area $F$ overlaps\\ 
\bottomrule
\end{tabular}
\caption{Notation used throughout the paper.}
\label{table:1}
}
\vspace{-0mm}
\end{table}

\begin{algorithm}[!t]
{\small
{\bf Setting:} A set $\cP$ of $N$ processes positioned at distinct coordinates. Each process can communicate with all other processes and knows their coordinates.  There are $M\geq 1$ identical fault areas $F$. The diameter of a fault area is $D$. The locations of any area $F$ is not known to correct processes. Each process covered by any $F$ is Byzantine.\\
{\bf Input:} Each process has initial value either 0 or 1.\\
{\bf Output:} Each correct process outputs decision subject to Geoconsensus.\\
{\bf \em Procedure for process $p_k\in \cP$}\\
// leaders selection \\
Let $P_D \leftarrow \emptyset$, $P_C  \leftarrow \cP$;  \\  
\While{$P_C \neq \emptyset$}{
    let $P_3 \subset P_D$ be a set of processes such that $\forall p_j \in P_3$, 
        $\mathit{Nb}(p_j, D)$ has distance $D$ independent set of at most 3; \\
    let $p_i \in P_3$, located in $(x_i, y_i)$ be
    the lexicographically smallest process in $P_3$, i.e. $\forall p_j \neq p_i \in P_3:$ located in $(x_j, y_j)$ either  $x_i<x_j$ or $x_i=x_j$ and $y_i<y_j$; \\
    add $p_i$ to $P_D$; \\
    remove $p_i$ from $P_C$; \\
   $\forall p_j \in \mathit{Nb}(p_i, D)$ remove $p_j$ from $P_C$; \\
}

// consensus \\
\eIf{$p_k \in P_D$}
{run \emph{PSL} algorithm, achieve decision $v$, broadcast $v$, output $v$;}
{
wait for messages with identical decision $v$ from at least $2M + 1$ processes from $\cP_D$, output $v$; 
}

\caption{\emph{BASIC} geoconsensus algorithm.}
\label{algorithm:simple}
}
\end{algorithm}	

\begin{algorithm}[!t]
{\small
{\bf Setting:} A set $\cP$ of $N$ processes positioned at distinct planar coordinates. Each process can communicate with all other processes and knows the coordinates of all other processes. The processes covered by the fault area $F$ at unknown location are Byzantine. There are $M\geq 1$ of identical fault areas $F$ and processes know $M$.\\
{\bf Input:} Each process has initial value either 0 or 1.\\
{\bf Output:} Each correct process outputs decision subject to geoconsensus\\
{\bf \em Procedure for process $p_k$}\\
// leaders selection \\
Compute the set $\cA$ of covers $A_j(R_i)$;\\
// {\bf For each} cover $A_j(R_i)\in \cA$ {\bf do}\\
 
$\cP_{min}\leftarrow$ a set of processes with minimum $y$-coordinate among covered by $A_j(R_i)$;\\
\If{$|\cP_{min}|=1$}
{
$l_j(A_j(R_i))\leftarrow$ the only process in $\cP_{min}$;\\
}
\Else
{
$l_j(A_j(R_i))\leftarrow$ the process in  $\cP_{min}$ with minimum $x$-coordinate;\\
}

Let $P_L$ be the set of leaders, one for each $A_j(R_i)\in \cA$; \\
// consensus \\
\eIf{$p_k \in \cP_L$}
{run \emph{PSL} algorithm, achieve decision $v$, broadcast $v$, output $v$}
{
wait for messages with identical decision $v$ from at least $2M + 1$ processes from $\cP_L$, output $v$ 
}
\caption{\emph{GENERIC} geoconsensus algorithm.}
\label{algorithm:consensus}
}
\end{algorithm}	

\noindent{\bf Proof of Theorem \ref{theorem:impossibility}.}

\begin{proof}
Set $N=3\cdot \kappa$, for some positive integer $\kappa\geq 1$. Place three areas $A$ on the plane in arbitrary locations.  To embed processes in $\cP$, consider a bijective placement function $f:\cP\rightarrow\cA$ such that $\kappa$ processes are covered by each area $A$. Let $v$ and $v'$ be two distinct input values $0$ and $1$. Suppose one area $A$ is fault area, meaning that all $\kappa$ processes in that area are faulty. 

This construction reduces the Byzantine goeconsensus problem to the impossibility construction for the classic Byzantine consensus problem given in the theorem in Section 4 of Pease {\it et al.}~\cite{Pease80} for the $3\kappa$ processes out of which $\kappa$ are Byzantine.
\qed
\end{proof}

\noindent{\bf Proof of Theorem \ref{theorem:square-cover}.}

\begin{proof}
The proof is to show that SQUARE-COVER is equivalent to the BOX-COVER problem which was shown to be NP-Complete by Fowler {\it et al.}~\cite{Fowler81}.
BOX-COVER is defined as follows: There is a set of $N$ points on the plane such that each point has unique integer coordinates. A closed box (rigid but relocatable) is set to be a square with side 2 and is axis-aligned. The problem is to decide whether a set of $k\geq 1$ identical axis-aligned closed boxes are enough to completely cover all $N$ points. 
Fowler {\ti et al.} provided a polynomial-time reduction of 3-SAT to BOX-COVER such that $k$ boxes will suffice if and only if the 3-SAT formula is satisfiable.  In this setting,  SQUARE-COVER (Definition \ref{definition:square-cover}) reduces to BOX-COVER for $\ell=2$. Therefore, the NP-Completeness of BOX-COVER extends to SQUARE-COVER. \qed
\end{proof}

\noindent{\bf Proof of Lemma \ref{lemma:slab-overlap}.}

\begin{proof}
It is sufficient to prove this lemma for adjacent slabs. Suppose slabs $R_i$ and $R_j$ are adjacent, i.e., $j=i+1$. According to the operation of \emph{GSQUARE}, after the location of $R_i$ is selected, only processes that are not covered by the slabs so far are considered for the selection of $R_j$. The first such process lies above the top (horizontal) side of $R_i$. Hence, there is a gap between the top side of $R_i$ and the bottom side of $R_j$.\qed
\end{proof}

\newpage
\noindent{\bf Proof of Lemma \ref{lemma:square-overlap}.}

\begin{proof}
It is sufficient to prove the lemma for adjacent squares. Suppose $A_j({R_i})$ and $A_k({R_i})$ are adjacent, i.e., $k=j+1$. Consider the operation of \emph{GSQUARE} in slab $R_i$ covered by $A_j({R_i})$ and $A_k({R_i})$. Area $A_k({R_i})$ only covers the processes that are not covered by $A_j({R_i})$ and, therefore, to the right of the right side of $A_j({R_i})$. As the left side of $A_k({R_i})$ is placed on the first such process, there is a non-empty gap between the two squares: $A_j({R_i})$ and $A_k({R_i})$.\qed  
\end{proof}

\noindent{\bf Proof of Lemma \ref{lemma:optimal-slab}.}

\begin{proof}
Notice that slab $R_i$ has height $height(R_i)=\ell$ which is the same as the sides of (axis-aligned) squares $A_j(R_i)$ used to cover $R_i$.

\emph{GSQUARE} operates such that it places a square $A$ so that some process $p$ lies on the left side of this square. Consider a sequence of such processes: $\sigma \equiv \langle p_1 \cdots p_u, p_{u+1} \cdots p_j \rangle$.   Consider any pair of subsequent processes $p_u$ and $p_{u+1}$ in $\sigma$ with respective coordinates $(x_u, y_u)$
and $(x_{u+1}, y_{u+1})$. \emph{GSQUARE} covers them with
non-overlapping squares with side $\ell$. Therefore, $x_u + \ell < x_{u+1}$. That is, the distance between consequent processes in $\sigma$ is greater than $\ell$. Hence, any such pair of processes may not be covered by a 
single square. Since the number of squares placed by \emph{GSQUARE} in slab $R_i$ is $k$, the number of processes in $\sigma$ is also $k$. Any algorithm that covers these processes with axis-aligned squares requires at least $k$ squares. \qed
\end{proof}

\noindent{\bf Proof of Lemma \ref{lemma:odd-even}.}

\begin{proof}
Consider two slabs $R_i$ and $R_{i+2}$ for $i\geq 1$. Consider a square $A_j(R_i)$ placed by \emph{GSQUARE}. Consider also two processes $p\in R_i$ and $p'\in R_{i+2}$, respectively.  The distance between $p$ and $p'$ is $d(p,p')>\ell$. 
Therefore, if $A_j(R_i)$  covers $p$, then it cannot cover $p'\in R_{i+2}$.  Therefore, no algorithm can produce the optimal number of squares $k_{opt}(\cR)$ less than the maximum between $k(\cR_{odd})$ and $k(\cR_{even})$. 
\qed
\end{proof}

\noindent{\bf Proof of Lemma \ref{lemma:square-approximation}.}

\begin{proof}
From Lemma \ref{lemma:optimal-slab}, we obtain that \emph{GSQUARE} is optimal for each slab $R_i$. From Lemma \ref{lemma:odd-even}, we get that for any algorithm $k_{opt}(\cR)\geq \max\{k(\cR_{odd}),k(\cR_{even})\}.$
Moreover, the \emph{GSQUARE} produces the total number of squares $k_{greedy}(\cR)=k(\cR_{odd})+k(\cR_{even}).$
Comparing $k_{greedy}(\cR)$ with $k_{opt}(\cR)$, we get
$$\frac{k_{greedy}(\cR)}{k_{opt}(\cR)}\leq \frac{k(\cR_{odd})+k(\cR_{even})}{\max\{k(\cR_{odd}),k(\cR_{even})\}}\leq \frac{2\cdot \max\{k(\cR_{odd}),k(\cR_{even})\}}{\max\{k(\cR_{odd}),k(\cR_{even})\}} \leq 2. ~~~~~ \qed$$
\end{proof}

\noindent{\bf Proof of Lemma \ref{lemma:circle-approximation}.}
\begin{proof}
We first show that 
$k^C_{opt}(\cR)\geq \max\{k^S(\cR_{odd}),k(\cR^S_{even})\},$
where $k^S(\cR_{odd})$ and $k^S(\cR_{even})$, respectively, are the number of squares $A=\ell\times \ell$ to cover the slabs in $\cR_{odd}$ and   $\cR_{even}$. Consider any square cover $A_j(R_i)$ of any slab $R_i$. A circle $C$ of diameter $\ell$ can cover at most the processes in $A_j(R_i)$ but not in any other square $A_l(R_i)$. This is because the perimeter of $C$ needs to pass through the left side of $A_j(R_i)$ (since there is a process positioned on that line in $A_j(R_i)$) and with diameter $\ell$, the perimeter of $C$ can touch at most the right side of $A_j(R_i)$.  

We now prove the upper bound. Since one square area $A=\ell\times \ell$ is now covered using at most 4 circles $C$ of diameter $\ell$, \emph{GCIRCLE} produces the total number of circles $k^C_{greedy}(\cR)=4\cdot (k^S(\cR_{odd})+k^S(\cR_{even})).$

Comparing $k^C_{greedy}(\cR)$ with $k^C_{opt}(\cR)$ as in  Lemma \ref{lemma:square-approximation}, we have that 
$$\frac{k^C_{greedy}(\cR)}{k^C_{opt}(\cR)}\leq \frac{4\cdot (k^S(\cR_{odd})+k^S(\cR_{even}))}{\max\{k^S(\cR_{odd}),k^S(\cR_{even})\}}\leq \frac{8\cdot \max\{k^S(\cR_{odd}),k^S(\cR_{even})\}}{\max\{k^S(\cR_{odd}),k^S(\cR_{even})\}} 
\leq 
8.~\qed$$
\end{proof}

\noindent{\bf Proof of Lemma \ref{lemma:overlap-A}.}

\begin{proof}
Suppose $F$ is axis-aligned.  $F$ may overlap at most two squares $A$ horizontally. Indeed, the total width covered by two squares in $\cA$ is $>2\ell$ since the squares do not overlap. Meanwhile, the total width of $F$ is $\ell$. 
Similarly, $F$ may overlap at most two squares vertically. Combining possible horizontal and vertical overlaps, we obtain that $F$ may overlap at most 4 distinct axis-aligned areas $A$. See Figure~\ref{fig:overlap-aaligned} for illustration.

\begin{figure}[!t]
\vspace{-0mm}
\begin{center}
\includegraphics[height=1.0in]{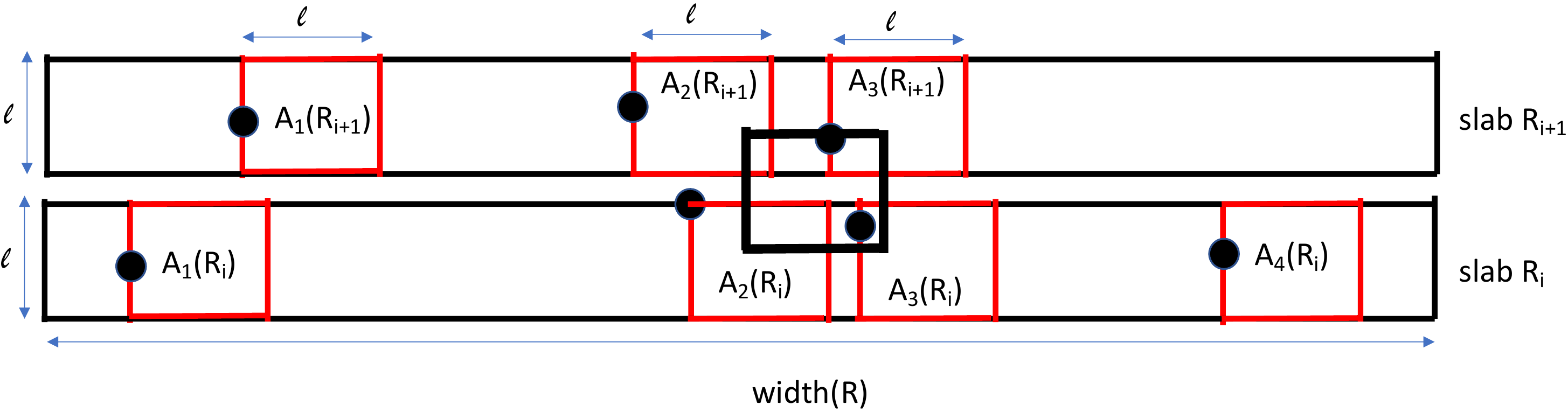}
\end{center}
\vspace{-7mm}
\caption{The maximum overlap of an axis-aligned fault area $F$ with the identical axis-aligned cover squares $A$ of same size.}
\label{fig:overlap-aaligned}
\vspace{-0mm}
\end{figure}

\begin{figure}[!t]
\vspace{-2mm}
\begin{center}
\includegraphics[height=1.15in]{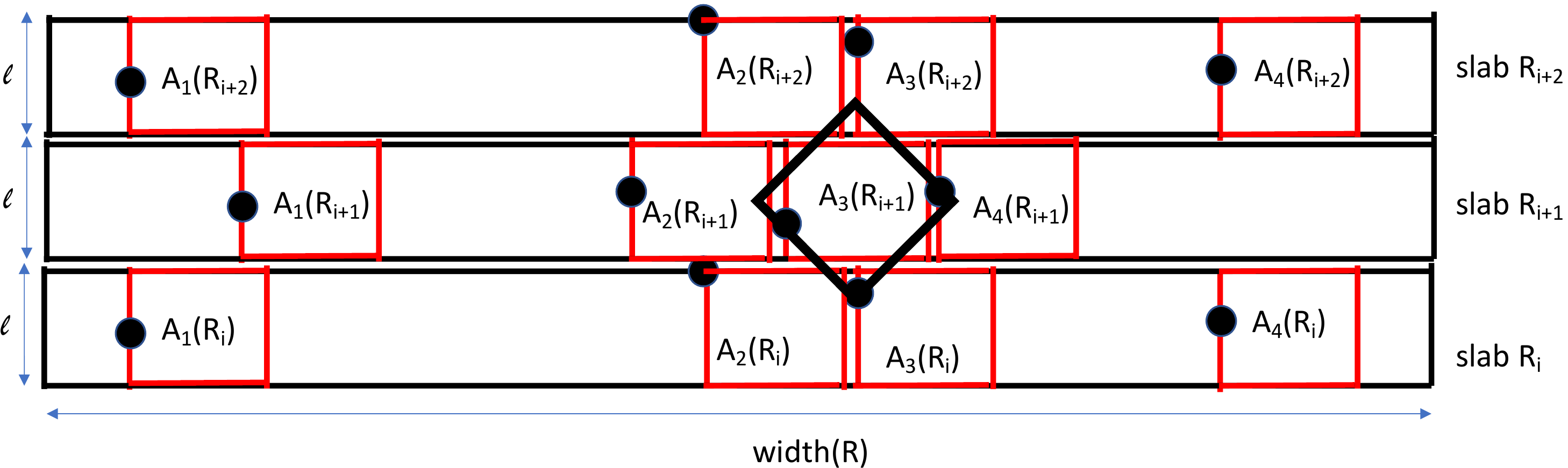}
\end{center}
\vspace{-7mm}
\caption{The maximum overlap of a non-axis-aligned fault area $F$ with the identical axis-aligned cover squares $A$ of the same size.}
\label{fig:overlap-nonaligned}
\vspace{-5mm}
\end{figure}

Consider now that $F$ is not axis-aligned. 
$F$ can span at most $\sqrt{2}\ell$ horizontally and $\sqrt{2}\ell$ vertically. Therefore, horizontally, $F$ can overlap at most three areas $A$.  Vertically, $F$ can overlap three areas as well. However, not all three areas on the top and bottom rows can be overlapped at once. Specifically, not axis-aligned $F$ can only overlap  2 squares in the top row and 2 in the bottom row. Therefore, in total, $F$ may overlap at most 7 distinct axis-aligned areas. Figure~\ref{fig:overlap-nonaligned} provides an illustration. 

For the case of circular $F$, one square area $A$ can be completely covered by 4 circles $C$. Furthermore, square $F$ of size $\ell$ overlaps at most 7 square areas $A$ of side $\ell$. Moreover, the circular $F$ of diameter $\ell$ can be inscribed in a square of side $\ell$. Therefore, a circular $F$ cannot overlap more than 7 squares, and hence  the circular $F$ may overlap in total at most $7\times 4=28$ circles $C$.    
\qed  
\end{proof}

\noindent{\bf Proof of Lemma \ref{lemma:overlap-Aroot2}.}
\begin{proof}
$F$ can extend, horizontally and vertically, at most $\sqrt{2}\cdot \ell/\sqrt{2}=\ell.$ Therefore, $F$ can overlap no more than two squares $A$ horizontally and two squares $A$ vertically. 

For the case of circular $F$ of diameter $\ell/\sqrt{2}$, it can be inscribed in a square of side $\ell/\sqrt{2}$. 
This square can overlap no more than $4$ squares of $\ell\times\ell$. Each such square can be covered by at most $4$ circles of diameter $\ell$. Therefore, the total number of circles to overlap the circular fault area $F$ is $4\times4 = 16$.
\qed
\end{proof}

\noindent{\bf Proof of Lemma \ref{lemma:overlap-circle}.}
\begin{figure}[!t]
\vspace{-0mm}
\begin{center}
\includegraphics[height=1.12in]{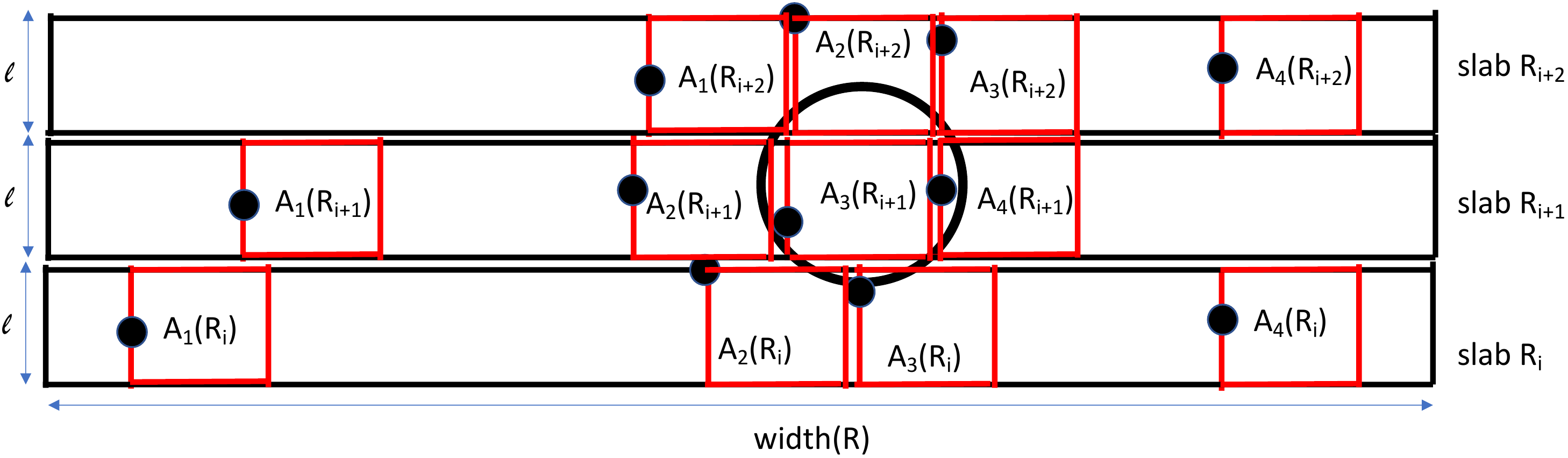}
\end{center}
\vspace{-6mm}
\caption{The maximum overlap of a circular fault area $F$ of diameter $\sqrt{2}\ell$ with axis-aligned cover squares $A$ of side $\ell$.}
\label{fig:overlap-circle}
\vspace{-2mm}
\end{figure}

\begin{proof}
Since $F$ is a circle of diameter $\sqrt{2}\ell$, $F$ can span horizontally and vertically at most $\sqrt{2}\cdot \ell$. 
Arguing similarly as in Lemma \ref{lemma:overlap-A}, $F$ can overlap either at most 3 squares $A$ in top row or 3 on the bottom row. Interestingly, if $F$ overlaps 3 squares in the top row, it can only overlap at most 2 in the bottom row and vice-versa. Therefore, in total, $F$ overlaps at most 8 distinct squares of side $\ell$. Figure~\ref{fig:overlap-circle} provides an illustration. 

Since one square of side $\ell$ can be completely covered using 4 circles of diameter $\ell$, $F$ of diameter $\sqrt{2}\ell$ can cover at most $8\times 4=32$ circles $C$ of diameter $\ell$. \qed
\end{proof}

\begin{samepage}

\noindent{\bf Proof of Theorem \ref{theorem:multiple-fault-main}.}

\begin{proof}
The proof for the case of $M=1$ extends to the case of $M>1$ as follows. Theorem \ref{theorem:single-fault-main} gives the bounds $f\leq N-\gamma$ and $|\cA|\geq \delta$ for one fault area for some positive integers $\gamma,\delta$. For $M$ fault areas, $M$ separate $|\cA|$ sets are needed, with each set tolerating a single fault area $F$. Therefore, the bounds of Theorem~\ref{theorem:single-fault-main} extend to multiple fault areas with a factor of $M$, i.e., \emph{GENERIC} needs $M\cdot \delta$ covers and $f\leq N-M \cdot \gamma$ faulty processes can be tolerated.   
Using the appropriate numbers from Theorem~\ref{theorem:single-fault-main} provides the claimed bounds.
\qed
\end{proof}

\noindent{\bf Proof of Theorem \ref{theorem:multiple-fault-main-circle}.}

\begin{proof}
For the first case, we have that $n(F)=28$, when cover set $\cA$ is of circles of diameter $\ell$ and the fault area $F$ is also a circle of diameter $\ell$. Therefore, when $|\cA|\geq 85M$, we have that at least $|\cA|-n(F)\geq 57M$ circles containing only correct processes. Since Algorithm \ref{algorithm:consensus} reaches consensus using only the values of the leader processes in each area $A$, when we have $|\cA|\geq 85M$, it is guaranteed that $\geq 2\cdot |\cA|/3+1\geq 2\cdot n(F)M+1$ leader processes are correct and hence \emph{GENERIC} can reach consensus. The fault tolerance guarantee of $f\leq N-57M$ can be shown similarly to the proof of Theorem \ref{theorem:single-fault-main}.

For the second result, we have shown that $n(F)=32$. Therefore, we need $|\cA|\geq 3\cdot n(F)+1\geq 97$ for one faulty circle $F$ of diameter $\sqrt{2}\ell$. For $M$ faulty circles, we need $|\cA|\geq 97M$. Therefore, the fault tolerance bound is $f\leq N-(2\cdot n(F)M+1)=N-65M$. 

For the third result, we have shown that $n(F)=16$ for a single faulty circle of diameter $\ell/\sqrt{2}$. Therefore, we need $|\cA|\geq 49M$ and $f\leq N-33M$. 
\qed
\end{proof}
\end{samepage}

\noindent{\bf Description of extension of \emph{GENERIC} to 3-d space in Section \ref{section:extension}.}

\ \\
Suppose the coordinates of process $p_i\in \cP$ are $(x_i,y_i,z_i)$. 
\emph{GCUBE} operates as follows. It first finds $x_{min},y_{min},z_{min}$ as well as $x_{max},y_{max},z_{max}$. Then, a smallest axis-aligned (w.r.t. $x$-axis) cuboid, i.e. rectangular parallelepiped,  $R$ with the left-bottom-near corner $(x_{min},y_{min},z_{min})$ and the right-top-far corner at $(x_{max},y_{max},z_{max})$ is constructed such that $R$ covers all $N$ processes in $\cP$. Assume that $z-axis$ is away from the viewer. The depth of $R$ is $depth(R)=z_{max}-z_{min}$; $width(R)$ and $height(R)$ are similar to \emph{GSQUARE}.    

\emph{GCUBE} now divides $R$ into a set $\cR$ of $m$ cuboids $\cR=\{R_1,\cdots,R_m\}$ such that $depth(R_i)=\ell$ but the $width(R_i)=width(R)$ and   $height(R_i)=height(R)$. 
Each $R_i$ is further divided into a set of $\cR_i$ of $n$ cuboids $\cR_i=\{R_{i1},\ldots,R_{in}\}$ such that each $R_{ij}$ has $width(R_{ij})=width(R)$ but $height(R_{ij})=\ell$ and $depth(R_{ij})=\ell$. 
Each cuboid $R_{ij}$ is similar to the slab $R_i$ shown in Figure \ref{fig:slab} but has depth $\ell$. 

It now remains to cover each axis-aligned cuboid $R_{ij}$ with cubic areas $A$ of side $\ell$. Area $A$ can be put on $R_{ij}$ so that the top left corner of $A$ overlaps with the top left corner of cuboid $R_{ij}$. Slide $A$ on the $x$-axis to the right so that there is a process covered by $R_{ij}$ positioned on the left vertical plane of $A$. Fix that area $A$ as one cover cube and name it $A_1(R_{ij})$. Now consider only the processes in $R_{ij}$ not covered by $A_1(R_{ij})$. Place another $A$ on those processes so that there is a point in $R_{ij}$ positioned on the left vertical plane of $A$ and there is no process on the left of $A$ that is not covered by $A_1(R_{ij})$. Let that $A$ be $A_2(R_{ij})$. Continue this way to cover all the processes in $R_{ij}$.  

Apply the procedure of covering $R_{ij}$ to all  $m\times n$ cuboids.
Lemma \ref{lemma:slab-overlap} can be extended to show that no two cuboids $R_{ij}, R_{kl}$ overlap. Lemma \ref{lemma:square-overlap} can be extended to show that no two cubic covers $A_o(R_{ij})$ and $A_p(R_{kl})$ overlap. For each cuboid $R_{ij}$, Lemma \ref{lemma:optimal-slab} can be extended to show that no other algorithm produces the number of cubes $k'(R_{ij})$ less than the number of cubes $k(R_{ij})$ produced by algorithm \emph{GCUBE}.

Since the cover for each square cuboid $R_{ij}$ is individually optimal, 
let $k_{opt}(\cR)$ be the number of axis-aligned cubes to cover all $N$  processes in $R$ in the optimal cover algorithm. We now show that $k_{greedy}(\cR)\leq 4\cdot k_{opt}(\cR)$, i.e., \emph{GCUBE} provides 4-approximation. 
We do this by combining two approximation bounds. The first is for the $m$ cuboids $R_i$, for which we show $2$-approximation. We then provide $2$-approximation for each cuboid $R_i$ which is now divided into $n$ cuboids $R_{ij}$. Combining these two approximations, we have, in total, a $4$-approximation.  

As in the 2-dimensional case, divide the $m$ cuboids in the set $\cR$ into two sets $\cR_{odd}$ snd $\cR_{even}$. Arguing as in Lemma \ref{lemma:optimal-slab}, we can show that $k_{opt}(\cR)\geq \max\{k(R_{odd}),k(R_{even})\}$ and $k_{greedy}(\cR)=k(R_{odd})+k(R_{even})$. 
Therefore, the ratio $k_{greedy}(\cR)/k_{opt}(\cR)\leq 2$ while dividing $R$ into $m$ cuboids.

Now consider any cuboid $R_i\in \cR_{odd}$ ($R_i\in \cR_{even}$ case is analogous). 
$R_i$ is divided into a set $\cR_i$ of $n$ cuboids $R_{ij}$. Divide $n$ cuboids in the set $\cR_i$ into two sets $\cR{i,odd}$ and $\cR{i,even}$ based on odd and even $j$. Therefore, it can be shown that, similarly to Lemma \ref{lemma:optimal-slab}, that  $k_{opt}(\cR_i)\geq \max\{k(R_{i,odd}),k(R_{i,even})\}$ and $k_{greedy}(\cR_i)=k(R_{i,odd})+k(R_{i,even})$.  Therefore, $k_{greedy}(\cR_i)/k_{opt}(\cR_i)\leq 2$. Combining the $2$-approximations each for the two steps, we have the overall $4$-approximation.

Let us now discuss the $16$-approximation for spheres of diameter $\ell$. One cube $A_l(R_{ij})$ of side $\ell$ can be completely covered by $4$ spheres of diameter $\ell$. Since, for cubes, \emph{GCUBE} is  $4$-approximation, we, therefore, obtain that \emph{GSPHERE} is a $16$-approximation. 
We omit this discussion but it can be shown that \emph{GSPHERE}, appropriately extended from \emph{GCIRCLE} into 3-dimensions, achieves $(2^{d-1}\cdot d^d)=4\cdot 27=108$ approximation.  

Now we need to find the overlap number $n(F)$.  Cube $A$ of side $\ell$ has diameter $D=\sqrt{3}\ell$. That means that a cubic fault area $F$ that has the same size as $A$ can overlap at most 3 cubes $A_l(R_{ij})$ in all 3 dimensions.  Therefore, $F$ can cover at most $3^3=27$  cubes $A_l(R_{ij})$.
For sphere $F$ of diameter $\ell$, since one cube $A_l(R_{ij})$ can be completely covered by $4$ spheres of diameter $\ell$ and $F$ can be inscribed inside $A_l(R_{ij})$, it overlaps the total $4\cdot 27=108$ spheres $A_l(R_{ij})$. For the the axis-aligned case of cubic fault area $F$, it can be shown that $n(F)=8$ cubes $A_l(R_{ij})$. This is because it can overlap with at most $4$ cubes $A_l(R_{ij})$ as Figure ~\ref{fig:overlap-aaligned} and, due to depth $\ell$, it can go up to two layers, totaling 8. $n(F)=32$ for sphere $F$ is immediate since each cube $A_l(R_{ij})$ is covered by $4$ spheres of diameter $\ell$, sphere of diameter $\ell$ can be inscribed inside a cube $A_l(R_{ij})$ of side $\ell$, and a faulty cube $F$ of side $\ell$ can overlap at most 8 axis-aligned cubes $A_l(R_{ij})$.
}
\end{document}